\documentclass{scrartcl}

\usepackage[english]{babel}

\usepackage{authblk}
\usepackage{amsthm}
\usepackage{amsmath}
\usepackage{amssymb}
\usepackage{mathtools}
\usepackage{graphicx}
\usepackage{hyperref}
\hypersetup{colorlinks,allcolors=blue}
\usepackage{cleveref}
\usepackage{subfigure}
\usepackage{algorithm}%
\usepackage{algorithmicx}%
\usepackage{algpseudocode}%

\usepackage{biblatex}
\AtEveryBibitem{
    \clearfield{urlyear}
    \clearfield{urlmonth}
    \clearfield{url}
}

\usepackage{csquotes}

\usepackage{tikz}
\usetikzlibrary{calc, arrows, patterns, angles, quotes, decorations.pathreplacing}
\tikzstyle{vertex}=[circle, fill, inner sep=0pt, minimum size=6pt]
\newcommand{\vertex}{\node[vertex]}

\usepackage{thmtools}
\usepackage{thm-restate}

\theoremstyle{remark}

\theoremstyle{plain}
\declaretheorem[name=Theorem]{thm}
\newtheorem{corollary}[thm]{Corollary}
\newtheorem{theorem}[thm]{Theorem}
\newtheorem{lemma}[thm]{Lemma}

\let\given\givenbase

\newcommand{\poly}{\operatorname{poly}}
\def\prob{\ensuremath\mathbb{P}}
\def\expect{\ensuremath\mathbb{E}}
\newcommand*\dd{\mathop{}\!\mathrm{d}}
\def\X{\ensuremath\mathcal{X}}
\def\Y{\ensuremath\mathcal{Y}}

\newcommand{\pder}[2][]{\frac{\partial#1}{\partial#2}}
\newcommand{\der}[2][]{\frac{\dd#1}{\dd#2}}

\let\originalleft\left
\let\originalright\right
\renewcommand{\left}{\mathopen{}\mathclose\bgroup\originalleft}
\renewcommand{\right}{\aftergroup\egroup\originalright}

\title{Improved Smoothed Analysis of 2-Opt for the Euclidean TSP}
\author{Bodo Manthey\thanks{\url{b.manthey@utwente.nl}} }
\author{Jesse van Rhijn\thanks{\url{j.vanrhijn@utwente.nl}. Corresponding author. Supported by NWO grant OCENW.KLEIN.176.}}
\affil{Department of Applied Mathematics, University of Twente, Enschede, The Netherlands}

\date{}

\begin{document}

\maketitle

\begin{abstract}
The 2-opt heuristic is a simple local search heuristic for the Travelling Salesperson Problem (TSP). 
Although it usually performs well in practice, its worst-case running time is exponential in
the number of cities. 
Attempts to reconcile this difference between practice and theory
have used smoothed analysis, in which adversarial instances are perturbed probabilistically. 
We are interested in the classical model of smoothed analysis for the Euclidean TSP, in which the perturbations are Gaussian. 
This model was previously used by Manthey \& Veenstra, who obtained smoothed complexity bounds polynomial in $n$,
the dimension $d$, and the perturbation strength $\sigma^{-1}$. 
However, their analysis only works for $d \geq 4$.
The only previous analysis for $d \leq 3$ was performed by Englert, R\"oglin \& V\"ocking, 
who used a different perturbation model which can be translated to Gaussian perturbations. 
Their model yields bounds polynomial in $n$ and $\sigma^{-d}$, and super-exponential in $d$. 
As the fact that no direct analysis exists for Gaussian perturbations that yields polynomial
bounds for all $d$ is somewhat unsatisfactory, we perform this missing analysis.
Along the way, we improve all existing smoothed complexity bounds for Euclidean 2-opt with Gaussian perturbations.
\end{abstract}

\textbf{Keywords:} Travelling Salesperson Problem, local search, smoothed analysis

\section{Introduction}

The Travelling Salesperson problem is a standard combinatorial optimization problem, which
has attracted considerable interest from academic, educational and industrial directions.
It can be stated rather compactly: given a Hamiltonian graph $G = (V, E)$ and edge weights
$w: E \to \mathbb{R}$, find a minimum weight Hamiltonian cycle (tour) on $G$.

Despite this apparent simplicity, the TSP is NP-hard \cite{korteCombinatorialOptimizationTheory2000}.
A particularly interesting variant of the TSP is the Euclidean TSP, in which
the $n$ vertices of the graph are identified with a point cloud in $\mathbb{R}^d$,
and the edge weights are the Euclidean distances between these points. Even this
restricted variant is NP-hard \cite{papadimitriouEuclideanTravellingSalesman1977}. 

As a consequence of this hardness, practitioners often turn to heuristics. One commonly used
heuristic is 2-opt \cite{aartsLocalSearchCombinatorial2003}. This heuristic takes as its input a tour $T$,
and finds two sets of two edges each, $\{e_1, e_2\} \subseteq T$ and $\{f_1, f_2\}
\nsubseteq T$, such that exchanging $\{e_1, e_2\}$ for $\{f_1, f_2\}$
yields again a tour $T'$, and the total weight of $T'$ is strictly
less than the total weight of $T$. This procedure is repeated with the new tour,
and stops once no such edges exist. The resulting tour is
said to be locally optimal.

Englert, R\"oglin and V\"ocking constructed Euclidean TSP instances on which 
2-opt can take exponentially many steps to find a locally optimal tour \cite{englertWorstCaseProbabilistic2014}.
Despite this pessimistic result, 2-opt performs remarkably well in practice,
usually requiring time sub-quadratic in $n$ and obtaining tours
which are only a few percent worse than the optimum \cite[chapter 8]{aartsLocalSearchCombinatorial2003}.

To explain this discrepancy, the tools of probabilistic analysis have proved
useful
\cite{mantheySmoothedAnalysis2Opt2013,chandraNewResultsOld1999, englertSmoothedAnalysis2Opt2016,engelsAveragecaseApproximationRatio2009,englertWorstCaseProbabilistic2014}.
In particular, smoothed analysis, a hybrid framework between worst-case and average-case analysis,
has been successfully used in the analysis of 2-opt \cite{englertSmoothedAnalysis2Opt2016,englertWorstCaseProbabilistic2014,mantheySmoothedAnalysis2Opt2013}.
In the original version of this framework, the instances one considers are initially adversarial,
and then perturbed by Gaussians. The resulting smoothed time
complexity is then generally a function of the instance size $n$ and
the standard deviation of the Gaussian perturbations, $\sigma$.

Englert et al.\ obtained smoothed time complexity bounds for 2-opt on
Euclidean instances
by considering a more general model,
in which the points are chosen in the unit hypercube according to arbitrary probability densities.
The only restrictions to these densities are that (i) they are independent,
and (ii) they are all bounded from above by $\phi$. Their results
can be transferred to Gaussian perturbations roughly by
setting $\phi = \sigma^{-d}$, which yields a smoothed complexity that
is $O(\poly(n, \sigma^{-d}))$, ignoring factors depending only on $d$.

As the exponential dependence on $d$ is somewhat unsatisfactory, Manthey \& Veenstra
\cite{mantheySmoothedAnalysis2Opt2013} performed a simpler smoothed analysis yielding bounds polynomial
in $n$, $1/\sigma$, and~$d$. However, their analysis is limited to $d \geq 4$.
While polynomial bounds for all $d$ can be obtained by simply taking the result of
Englert et al.\ for $d \in \{2, 3\}$,
no smoothed analysis that directly uses Gaussian perturbations exists for these cases.
We set out to perform this missing analysis, improving the
smoothed complexity bounds for all $d \geq 2$ along the way.

Our analysis combines ideas from both Englert et al.\ and Manthey \& Veenstra.
From the former, we borrow the idea of conditioning on the outcomes of some of the
distances between points in an arbitrary 2-change. We can then analyze the 2-change
by examining the angles between certain edges in the 2-change, which are themselves random
variables. 
From the latter, we borrow the Gaussian perturbation model
(originally introduced by Spielman \& Teng for the Simplex Method
\cite{spielmanSmoothedAnalysisAlgorithms2004}).

We also note that in addition to improving the results of Manthey \& Veenstra,
our approach is significantly simpler than the analysis of Englert et al.
The crux of the simplification is a carefully constructed random experiment
to model a single 2-change, which allows us to bypass the need for
the involved convolution integrals used by Englert et al.

We will begin by introducing some definitions
and earlier results, before providing basic probability theoretical results
(\Cref{sec:prelims})
that we will make heavy use of throughout the paper. We then proceed by
analyzing a single 2-change in a similar manner as Englert et al.,
simplifying some of their analysis in the process
(\Cref{sec:single_2change}). Next, we prove
a first smoothed complexity bound by examining so-called linked pairs of
2-changes (\Cref{sec:linked_pairs}),
an idea used by both Englert et al.\ and Manthey \& Veenstra. Finally,
we improve on this bound for $d \geq 3$ (\Cref{sec:d3}), yielding the best known bounds
for all dimensions.

\section{Preliminaries}\label{sec:prelims}

\subsection{Travelling Salesperson Problem}

Let $\Y \subseteq [-1, 1]^d$ be a point set of size $n$. The Euclidean Travelling
Salesperson Problem (TSP) asks for a tour that visits each point $y \in \Y$
exactly once, such that the total length of the tour is minimized. The length of a tour
in this variant of the TSP is the sum of the Euclidean
distances between consecutive points in the
tour. Formally, if the points in $\Y$ are visited in the order
$T = (y_{\pi(i)})_{i=0}^{n-1}$ defined by
a permutation $\pi$ of $[n]$, then the length of the tour $T$ is
\[
    L(T) = \sum_{i=0}^{n-1} \|y_{\pi(i)} - y_{\pi(i+1)}\|,
\]
where the indices are taken modulo $n$, and $\|\cdot\|$ denotes the
standard Euclidean norm in $\mathbb{R}^d$. Since the Euclidean TSP is
undirected, the tour $T'$ in which the vertices are visited in the reverse order
has the same length as $T$. We consider these tours to be identical.

\subsection{Smoothed Analysis}

Smoothed analysis is a framework for the analysis of algorithms, which
was introduced in 2004 by Spielman \& Teng \cite{spielmanSmoothedAnalysisAlgorithms2004}.
The method is particularly suitable to algorithms with a fragile worst-case
input \cite{mantheySmoothedAnalysisLocal2021}. 
Since its introduction, the method has been applied
to a wide variety of algorithms
\cite{mantheySmoothedAnalysisAnalysis2011,spielmanSmoothedAnalysisAttempt2009}.

Heuristically, one imagines that an adversary chooses an input to the algorithm.
The input is then perturbed in a probabilistic fashion. The hope is that any
particularly pathological instances that the adversary might choose are destroyed by
the random perturbation. One then computes a bound on the expected number
of steps that the algorithm performs, where the expectation is taken with respect
to the perturbation.

For our model of a smoothed TSP instance, we allow the adversary to choose a
point set $\Y \subseteq [-1, 1]^d$ of size $n$. We then perturb each point $y_i \in \Y$
with an independent $d$-dimensional Gaussian random variable $g_i$, $i \in [n]$,
with mean 0 and standard deviation $\sigma$.
This yields a new point set, $\X = \{y_i + g_i \mid y_i \in \Y\}$.
We will bound
the expected number of steps taken by the 2-opt heuristic on the TSP
instance defined by $\X$, with
the expectation taken over this Gaussian perturbation.
We will refer to this quantity as the smoothed
complexity of 2-opt.

For the purposes of our analysis, we always assume that
$\sigma \leq 1$.
This is a mild restriction, as the bound for $\sigma=1$ also applies to all larger values of $\sigma$, and small perturbations are particularly interesting in smoothed analysis.

For a general outline of the strategy, consider a 2-change where the
edges $\{a, z_1\}$ and $\{b, z_2\}$ are replaced by
$\{a, z_2\}$ and $\{b, z_1\}$. The change in tour length of
this 2-change is
\[
    \Delta = \|a - z_1\| + \|b - z_2\| - \|a - z_2\| - \|b - z_1\|.
\]
Since the locations of the points $\{a, b, z_1, z_2\}$ are random variables,
so is $\Delta$. We seek to bound the probability that there exists a
2-change whose improvement is exceedingly small, enabling us to use a potential
argument. 

Let $\Delta_\mathrm{min}$ denote
the improvement of the least-improving 2-change in the instance. If $\prob(\Delta_\mathrm{min}\leq \epsilon)$
is suitably small for small $\epsilon$, then each iteration is likely to decrease the tour length
by a large amount. As long as the initial tour has bounded length,
this then provides a limit to the number of iterations that the heuristic can
perform, since the tour length is bounded from below by 0.

\subsection{Basic Results}

We state some general results that we will need at points throughout the paper.

The next lemma provides a simple framework that we can use to prove smoothed
complexity bounds for 2-opt.

Let $\Delta_\mathrm{min}$ denote the smallest improvement of any 2-change,
and let $\Delta^\mathrm{link}_\mathrm{min}$ denote the smallest improvement
of any pair of linked 2-changes (see \Cref{sec:linked_pairs} for a definition of linked pairs).

\begin{lemma}[\mbox{\cite[Lemma 2.2]{mantheySmoothedAnalysis2Opt2013}}]\label{lemma:blackbox}
    Suppose that the longest tour has a length of at most
    $L$ with probability at least $1 - 1/n!$. Let $\alpha > 1$ be a constant.
    If for all $\epsilon > 0$ it holds that
    $\prob(\Delta_\mathrm{min} \in (0, \epsilon]) = O\left(P\epsilon^\alpha\right)$,
    then the smoothed complexity of 2-opt is bounded from above by
    $O(P^{1/\alpha} L)$. The same holds if 
    we replace $\Delta_\mathrm{min}$ by $\Delta_\mathrm{min}^\mathrm{link}$,
    provided that $P^{1/\alpha}L = \Omega(n^2)$.
\end{lemma}

\subsubsection{Probability Theory}

We provide some basic probability theoretical results.
Throughout the paper, given a random variable $X$, we denote
its probability density by $f_{X}$ and its cumulative distribution function
by $F_X$. If we furthermore condition on some event $Y$,
we write $f_{X|Y}$ for the conditional density of $X$ given
$Y$.

\subsubsection*{Chi Distributions}

Suppose we are given two points $y_1, y_2 \in \Y$ and perturb
both points with independent Gaussian random variables $g_1$ and $g_2$, resulting
in $x_i = y_i + g_i$, $i \in [2]$. Then the distance $\|x_1 - x_2\|$
between the two perturbed points
is distributed according to a noncentral $d$-dimensional chi distribution
with noncentrality parameter $s = \|y_1 - y_2\|$,
which we denote $\chi_d^s$. We call $\chi_d^0$ a central $d$-dimensional $\chi$
distribution.
We have two useful expressions for the  chi
distribution \cite{johnsonContinuousUnivariateDistributions1995}:
\begin{align}\label{eq:noncentral_chi}
    \chi_d^s(r) = \frac{e^{-\frac{r^2+s^2}{2\sigma^2}} \cdot \frac{r^{d-1}}{\sigma^d}}
        {(rs/\sigma^2)^{d/2-1}} I_{d/2 - 1}\left(\frac{rs}{\sigma^2}\right)
        = e^{-\frac{s^2}{2\sigma^2}} \sum_{i=0}^\infty \frac{1}{i!}\left(\frac{s^2}{2\sigma^2}\right)^i
            \chi_{d + 2i}(r),
\end{align}
where $\chi_d(r) = \chi_d^0(r)$, the central chi distribution.
Here, $I_\nu(x)$ denotes the modified Bessel function of
the first kind, of order $\nu > -1/2$, defined as \cite{abramowitzHandbookMathematicalFunctions1974}
\begin{align}\label{eq:bessel_definition}
    I_\nu(x) = \sum_{k=0}^\infty \frac{1}{k!\Gamma(k + \nu + 1)} \left(
        \frac{x}{2} 
    \right)^{2k + \nu}.
\end{align}

\subsubsection*{General Results}

In the following, we use the notion of stochastic dominance.
Let $X$ and $Y$ be two real-valued random variables. We say that
$X$ stochastically dominates $Y$ if for all $x$, it holds that
$\prob(X \geq x) \geq \prob(Y \geq x)$, and this inequality
is strict for some $x$. We may equivalently say that the density
of $X$ stochastically dominates the density of $Y$.

To use \Cref{lemma:blackbox}, we need to limit
the probability that any TSP tour in our smoothed instance
is too long. This was previously done
by Manthey \& Veenstra; we state their result in
\Cref{lemma:box}.

\begin{lemma}[\mbox{\cite[Lemma 2.3]{mantheySmoothedAnalysis2Opt2013}}]\label{lemma:box}
    Let $c \geq 2$ be a sufficiently large constant, and let
    $D = c \cdot (1 + \sigma \sqrt{n \log n})$. Then
    $\prob(\X \nsubseteq [-D, D]^d) \leq 1/n!$.
\end{lemma}

The next lemma is a reformulation of another result by Manthey \& Veenstra \cite{mantheySmoothedAnalysis2Opt2013}.
The lemma is very useful in conjunction with
\Cref{lemma:chi_expect}, as we will
have cause to condition on the outcome of drawing noncentral $d$-dimensional chi random
variables.

\begin{lemma}[\mbox{\cite[Lemma 2.8]{mantheySmoothedAnalysis2Opt2013}}]\label{lemma:chi_stochdom}
    The noncentral $d$-dimensional chi distribution with parameter $\mu > 0$
    and standard deviation $\sigma$
    stochastically dominates the central $d$-dimensional
    chi distribution with the same standard deviation.
\end{lemma}

The following lemma from Manthey \& Veenstra is slightly generalized compared
to its original statement. We do not provide a proof, since the original proof
remains valid when simply replacing the original assumption with ours.

\begin{lemma}[\mbox{\cite[Lemma 2.7]{mantheySmoothedAnalysis2Opt2013}}]\label{lemma:chi_expect}
    Assume $c \in \mathbb{R}_{\geq 0}$ is a fixed constant and $d \in \mathbb{N}$ is
    fixed and arbitrary with $d > c$. Let $\chi_d$ denote the $d$-dimensional
    chi distribution with variance $\sigma^2$. Then
    \[
        \int_0^\infty \chi_d(x) x^{-c} \dd x = \Theta\left(
            \frac{1}{d^{c/2}\sigma^c} 
        \right).
    \]
\end{lemma}

\subsection{Limiting the Adversary}

In our analysis we will closely study the angles between edges in the
smoothed TSP instance. These angles can be initially specified
to our detriment by the adversary. However, the power of the adversary
is limited by the strength of the Gaussian perturbations.
We quantify the power of the adversary in \Cref{thm:angle_distr_bound}.
See \Cref{fig:angle_distribution} for a sketch accompanying the theorem.

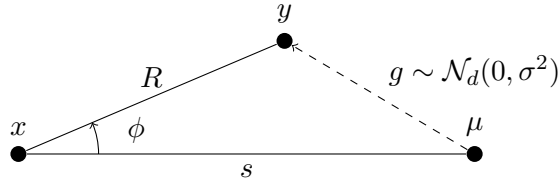
\begin{figure}[t]
\centering
\begin{tikzpicture}
    \vertex[label=$x$](x) at (0, 0) {};
    \vertex[label=$\mu$](mu) at (6, 0) {};
    \vertex[label=$\text{$y$}$](y) at (3.5, 1.5) {};

    \draw[black] (x) -- node[above] {$R$} ++ (y);
    \draw[black, dashed, ->] (mu) -- node[above, xshift=1.25cm] {$g \sim \mathcal{N}_d(0, \sigma^2)$} ++ (y);
    \draw[black] (x) -- node[below] {$s$} ++ (mu);
    
    \pic [draw, ->, "$\phi$", angle eccentricity=1.5, angle radius=30] {angle = mu--x--y};
\end{tikzpicture}
\caption{The setting of \Cref{thm:angle_distr_bound}. As mentioned in the proof of \Cref{thm:angle_distr_bound}, we may assume without loss of generality that $\mu$ lies on $L$.\label{fig:angle_distribution}}
\end{figure}

\begin{theorem}\label{thm:angle_distr_bound}
    Let $L$ be some line in $\mathbb{R}^d$, and let
    $x \in L$. Let $y$ be a point drawn from a $d$-dimensional
    Gaussian distribution with mean $\mu \in \mathbb{R}^d$
    and variance $\sigma^2$.
    Let $\phi$ denote the angle between $L$ and $x - y$,
    and let $R = \|x-y\|$ and $s = \|x - \mu\|$.
    Let $f_{\phi|R=r}$ denote the density of $\phi$, conditioned on a specific outcome
    $r > 0$ for $R$. Then for all $d \geq 2$,
    \[
        \sup_{\phi \in [0, \pi]} f_{\phi|R=r}(\phi) = O\left(\sqrt{d} + \frac{\sqrt{rs}}{\sigma}\right).
    \]
    Moreover, for $d \geq 3$,
    \[
        \sup_{\phi \in (0, \pi)} \frac{f_{\phi|R=r}(\phi)}{\sin\phi}
            = O\left(\sqrt{d} + \frac{rs}{\sigma^2\sqrt{d}}\right).
    \]
\end{theorem}

\Cref{thm:angle_distr_bound} yields the following corollary,
which provides information on the angle between two
Gaussian random points in $\mathbb{R}^d$ with respect to some third point.
This corollary is especially useful when analyzing 2-changes
in smoothed TSP instances.

\begin{corollary}\label{cor:angle_distr_bound}
    Let $x \in \mathbb{R}^d$.
    Let $y$ and $z$ be drawn from $d$-dimensional Gaussian distributions with
    arbitrary means and the same variance $\sigma^2$.
    Let $\phi$ denote the angle between $y - x$ and $z - x$, and let
    $R = \|x - y\|$ and $S = \|x - z\|$.
    Let $f_{\phi|R=r,S=s}$ denote the density of $\phi$ conditioned
    on some outcome $r > 0$ for $R$ and $s > 0$ for $S$. Then for
    all $d \geq 2$,
    \[
        \sup_{\phi \in [0, \pi]} f_{\phi|R=r,S=s}(\phi) = O\left(\sqrt{d} + \frac{\sqrt{\min\{r\bar{r}, s\bar{s}}\}}{\sigma}\right),
    \]
    where $\bar{r} = \|x - \expect(y)\|$ and $\bar{s} = \|x - \expect(z)\|$.
    Moreover, for $d \geq 3$,
    \[
        \sup_{\phi \in (0, \pi)} \frac{f_{\phi|R=r,S=s}(\phi)}{\sin\phi} 
            = O\left(\sqrt{d} + \frac{\min\{r\bar{r}, s\bar{s}\}}{\sigma^2\sqrt{d}}\right).
    \]
\end{corollary}

\begin{proof}[Proof (assuming \Cref{thm:angle_distr_bound})]
    We denote the density of $\phi$ conditioned on $R = r$ and $S = s$ by
    $f_{\phi|R=r,S=s}$. We perform a random experiment as follows. 
    
    If $r \leq s$, then we let an
    adversary determine the position of $z$, subject to $S = s$. Subsequently,
    we draw the line $L$ through $x$ and $z$.
    \Cref{thm:angle_distr_bound} then yields a bound for $f_{\phi|R=r,S=s}$ of
    $O(\sqrt{d} + \sqrt{r\bar{r}}/\sigma)$. The same process yields the bound for
    $f_{\phi|R=r,S=s}(\phi)/\sin\phi$ when $d \geq 3$.
    
    If $s \leq r$, then we use a similar argument, just swapping the roles
    of $y$ and $z$. This yields $O(\sqrt{d} + \sqrt{s\bar{s}}/\sigma)$.
    
    Combining these two bounds yields the corollary. 
\end{proof}

The remainder of this section is devoted to proving
\Cref{thm:angle_distr_bound}. Recall the formulas for
$\chi_d^s$, cf.\ \Cref{eq:noncentral_chi}. During the proof of
\Cref{thm:angle_distr_bound},
we will need to bound $\chi_d^s$ from below, for which
We require some lower bounds on $I_\nu$.
We thus spend some time in this
section proving such bounds.

The following bound on $I_\nu$ holds for all $x \geq 0$ and $\nu > -1/2$;
it results from keeping only the $k = 0$ term in \Cref{eq:bessel_definition}.
\begin{lemma}\label{lemma:besselbound}
    For all $x \geq 0$ and $\nu >-1/2$, 
    \[
        I_\nu(x) \geq \frac{(x/2)^\nu}{\Gamma\left(\nu + 1\right)}. 
    \]
\end{lemma}

As will become apparent during the proof of \Cref{thm:angle_distr_bound},
the bound in \Cref{lemma:besselbound} is too weak for large values of $x$.
We thus need a stronger bound for this regime.

\begin{lemma}\label{lemma:besselbound_largex}
    Given $x > 1$ and $\nu \geq 0$, it holds that
    \[
        I_\nu(x) \geq c_\nu \cdot
            \frac{e^x}{\sqrt{x}},
    \]
    for some $c_\nu > 0$ that depends only on $\nu$.
\end{lemma}

\begin{proof}
    First, suppose $\nu \geq 1/2$.
    Our starting point is the following integral representation
    of $I_\nu$, which holds for $\nu > -1/2$ \cite{abramowitzHandbookMathematicalFunctions1974}:
    \begin{align}\label{eq:bessel_integral_repr}
        I_\nu(x) = \frac{(x/2)^\nu}{\pi^{1/2}\Gamma(\nu+1/2)} 
            \int_{-1}^1 e^{xt}(1-t^2)^{\nu-\frac{1}{2}}\dd t.
    \end{align}
    Observe first that the factor in front of the integral is non-negative,
    as is the integrand. We first restrict the domain of integration
    to $(1-1/x, 1)$, which is permissible as $x > 1$. Next,
    we use the identity $(1-t^2) = (1-t)(1+t)$ to replace
    $(1-t^2)^{\nu-1/2}$ in the integrand by $(1-t)^{\nu-1/2}$.
    This yields a lower bound, since $t$ only takes positive values over
    the restricted domain of integration, and $\nu \geq 1/2$.
    
    Next, we substitute $u = 1-t$, which yields
    \[
        \int_0^{1/x} e^{x(1-u)}u^{\nu-\frac{1}{2}}\dd u
            = e^x \int_{0}^{1/x} e^{-xu}u^{\nu-\frac{1}{2}} \dd u
                \geq e^x \int_0^{1/x} (1-xu)u^{\nu-\frac{1}{2}} \dd u,
    \]
    making use of the standard inequality $e^x \geq 1 + x$. Note
    that the integrand remains non-negative for all values
    of $u$ over which we integrate.
    The remaining integral evaluates to
    \begin{align*}
       \int_0^{1/x} (1-xu)u^{\nu-\frac{1}{2}} \dd u
        &= \frac{1}{\nu+1/2} \frac{1}{x^{\nu+1/2}}
            - \frac{x}{\nu+3/2}\frac{1}{x^{\nu+3/2}} \\
        &= \left(\frac{1}{\nu + 1/2} - \frac{1}{\nu + 3/2}\right)x^{-\nu-1/2}.
    \end{align*}
    
    Thus, we are left with
    \[
        I_\nu(x) \geq \left(\frac{1}{\nu + 1/2} - \frac{1}{\nu + 3/2}\right) 
            \frac{1}{2^\nu\sqrt{\pi}\Gamma(\nu + 1/2)} \frac{e^x}{\sqrt{x}}.
    \]
    Letting $c_\nu$ be the entire prefactor of $e^x/\sqrt{x}$, we are done
    for $\nu \geq 1/2$.
    
    The case $\nu < 1/2$ can be carried out analogously; however,
    rather than using $1-t^2 = (1+t)(1-t) \geq 1 - t$, we instead
    use $1-t^2 = (1+t)(1-t) \leq 2(1-t)$, since $1-t^2$ now
    appears in the denominator of the integrand in \Cref{eq:bessel_integral_repr}.
\end{proof}

While \Cref{lemma:besselbound_largex} is useful for large values of
$x$ and constant $\nu$, it is too weak for large values of $\nu$ due to
the constant $c_\nu$. We can however
use it to obtain another bound, which we will use at a key step in the
proof of \Cref{thm:angle_distr_bound}. First, we
need the following lemma, which can be found as an equation in a paper by Amos.

\begin{lemma}[\cite{amosComputationModifiedBessel1974}]\label{lemma:bessel_ratio}
    For all $x > 0$ and $\nu \geq 1$,
    \[
        \frac{I_{\nu}(x)}{I_{\nu-1}(x)} \geq \frac{\sqrt{x^2 + \nu^2} - \nu}{x}.
    \]
\end{lemma}

We can use this lemma recursively to bound $I_\nu$ from below
for all $\nu \geq 0$,
with the base case given by \Cref{lemma:besselbound_largex}.

\begin{lemma}\label{lemma:besselbound_generic}
    There exists a constant $c > 0$ such that,
    for all $x > 1$ and $\nu \geq 0$,
    \[
        I_\nu(x) \geq c\cdot \left(\frac{\sqrt{x^2 + \nu^2} - \nu}{x}\right)^{\nu+\frac{1}{2}}
            \frac{e^{\sqrt{x^2 + \nu^2}}}{\sqrt{x}}.
    \]
\end{lemma}

\begin{proof}
    First, we assume $\nu \in \mathbb{N}$ for the sake of clarity; fractional
    $\nu$ and $\nu < 1$ will be addressed at the end of the proof.
    We start by using \Cref{lemma:bessel_ratio}.
    Applied iteratively, it yields
    \[
        I_\nu(x) \geq I_0(x) \prod_{k=1}^\nu \frac{\sqrt{x^2 + k^2} - k}{x}
            = x^\nu I_0(x) \prod_{k=1}^\nu \frac{1}{\sqrt{x^2 + k^2} + k}.
    \]
    Equivalently,
    \[
        \frac{I_0(x) \cdot x^\nu}{I_\nu(x)} \leq \prod_{k=1}^\nu (\sqrt{x^2 + k^2} + k). 
    \]
    To bound this product, we first take its logarithm to convert
    it to a sum:
    \[
        \ln \prod_{k=1}^\nu (\sqrt{x^2 + k^2} + k)
            = \sum_{k=1}^\nu \ln (\sqrt{x^2 + k^2} + k).
    \]
    
    It is tempting to now bound this sum by integrating the summand over
    $[1, \nu + 1]$, as the summand is monotone increasing in $k$.
    However, the resulting bound turns out to be slightly too weak for
    our purposes. Instead, we refine this by using the Euler-Maclaurin formula
    \cite{apostolElementaryViewEuler1999}.
    The formula states that, for a function $f$ that is $p$-times continuously
    differentiable on $[m, n]$,
    \[
        \sum_{i=m}^n f(i) = \int_m^n f(k)\dd k + \frac{f(n) + f(m)}{2}
            + \sum_{k=1}^{\lfloor p/2 \rfloor} \frac{B_{2k}}{(2k)!} (f^{(2k-1)}(n) - f^{(2k-1)}(m))
                + R_p,
    \]
    where $B_k$ denotes the $k$th Bernoulli number with $B_1 = \frac{1}{2}$,
    and $R_p$ is a remainder term. The remainder can be bounded from above as
    \cite{abramowitzHandbookMathematicalFunctions1974}
    \[
        |R_p| \leq \frac{2\zeta(p)}{(2\pi)^p}\int_m^n |f^{(p)}(x)|\dd x,
    \]
    with $\zeta$ the Riemann zeta function.
    We apply this formula to $f(k) = \ln(\sqrt{x^2 + k^2} + k)$. It suffices
    to take $p = 2$, so that we retain only the first term of the sum. We have
    \[
        f'(k) = \frac{1}{\sqrt{x^2 + k^2}}.
    \]
    Observe that $f''(k) \leq 0$ for all $x, k \in \mathbb{R}$, so we have
    $|f''(k)| = -f''(k)$. This enables us to write the estimate for the
    remainder term as
    \[
        |R_2| \leq -\frac{2\zeta(2)}{4\pi^2} \int_1^\nu f''(k)\dd k
            = - \frac{1}{12} \left(f'(\nu) - f'(1)\right).
    \]
    Since $B_{2} = \frac{1}{6}$ \cite{abramowitzHandbookMathematicalFunctions1974}, we obtain
    \begin{align*}
        \sum_{k=1}^\nu f(k)
            &= \int_1^\nu f(k) \dd k + \frac{f(1) + f(\nu)}{2}
                + \frac{1}{12}(f'(\nu) - f'(1))
                    + R_p \\
            &\leq  
            \int_1^\nu f(k) \dd k + \frac{f(1) + f(\nu)}{2}
                + \frac{1}{6}\left|
                    f'(\nu) - f'(1) 
                \right|.
    \end{align*}
    The integral evaluates to
    \[
        \sqrt{1+x^2} - \sqrt{x^2 + \nu^2}
            + \ln\left(
                \frac{1}{1 + \sqrt{1 + x^2}} 
            \right)
            + \nu\ln\left(
                \sqrt{x^2 + \nu^2} + \nu 
            \right).
    \]
    Meanwhile, we have
    \[
        \frac{f(1) + f(\nu)}{2} = \ln \sqrt{1 + \sqrt{1 + x^2}} + \frac{1}{2}\ln(\sqrt{x^2 + \nu^2} + \nu), 
    \]
    and
    \[
        \left|
            f'(1) - f'(\nu) 
        \right|
            = \frac{1}{\sqrt{x^2 + 1}} - \frac{1}{\sqrt{x^2 + \nu^2}} \leq 1.
    \]
    Putting this all together,
    \[
        \sum_{k=1}^\nu \ln(\sqrt{x^2 + \nu^2} + \nu)
            \leq \sqrt{1+x^2} - \sqrt{x^2 + \nu^2}
                + \ln\left(\frac{(\sqrt{x^2+\nu^2}+\nu)^{\nu+\frac{1}{2}}}{\sqrt{1 + \sqrt{1 + x^2}}}\right)
                + 1.
    \]

    Exponentiating, we find
    \begin{align*}
        \frac{I_0(x)x^\nu}{I_\nu(x)}
            \leq e\cdot \frac{e^{\sqrt{1+x^2} - \sqrt{x^2 + \nu^2}}}
                {\sqrt{1 + \sqrt{1 + x^2}}}
                \left(\sqrt{x^2 + \nu^2} + \nu\right)^{\nu+\frac{1}{2}}.
    \end{align*}
    Using that $1 + \sqrt{1 + x^2} \geq x$,
    \begin{align*}
        I_{\nu}(x) &\geq \frac{1}{e}\cdot \left(
            \frac{x}{\sqrt{x^2 + \nu^2} + \nu}
        \right)^{\nu + \frac{1}{2}}
        e^{\sqrt{x^2 + \nu^2} - \sqrt{1+x^2}} I_0(x) \\
        &= \frac{1}{e} \cdot \left(\frac{\sqrt{x^2 + \nu^2} - \nu}{x}\right)^{\nu+\frac{1}{2}}
        e^{\sqrt{x^2 + \nu^2} - \sqrt{1+x^2}} I_0(x).
    \end{align*}
    
    To conclude the proof for integral
    $\nu$, we apply \Cref{lemma:besselbound_largex} for
    $\nu = 0$ to obtain $I_0(x) \geq c_0 \cdot e^x /\sqrt{x}$,
    and observe that $|\sqrt{1 + x^2} - x| \leq 1$ for
    all $x \geq 0$.
    
    For fractional $\nu$, one can follow the same proof,
    simply replacing $I_0$ by $I_{\nu'}$ for some $\nu' \in (0, 1)$
    throughout. Meanwhile, for $\nu < 1$, one can choose a suitable constant
    to match the bound from the lemma statement to the bound from
    \Cref{lemma:besselbound_largex}.
\end{proof}

The final piece of preparation for \Cref{thm:angle_distr_bound} is
now the following inequality.

\begin{lemma}\label{lemma:simple_inequality}
    Let $x \geq 0$ and $y \geq 1$. Then
    \[
        \left(\frac{\sqrt{x^2 + y^2} + y}{\sqrt{x^2 + \left(y-\frac{1}{2}\right)^2}
            + \left(y - \frac{1}{2}\right)}\right)^y \leq e.
    \]
\end{lemma}

\begin{proof}
    Let $f(x,y)$ denote the function in brackets. We first
    show that $f$ is nonincreasing in $x$. Observe that
    $f(x,y)$ is nonincreasing if and only if $\ln f(x,y)$ is nonincreasing.
    We have
    \[
        \pder{x} \ln \left(\sqrt{x^2+y^2} + y\right) 
            = \frac{1}{\sqrt{x^2 + y^2} + y}
                \cdot \frac{x}{\sqrt{x^2+y^2}}.
    \]
    Thus,
    \begin{multline*}
        \pder{x} \ln f(x,y)
        = y \cdot x \\
        \times\left(
            \frac{1}{\sqrt{x^2+y^2} \left(\sqrt{x^2+y^2} + y^2\right)}
            - \frac{1}{\sqrt{x^2 + \left(y-\frac{1}{2}\right)^2}
                \left(\sqrt{x^2 + \left(y-\frac{1}{2}\right)^2}
                    + \left(y-\frac{1}{2}\right)\right)}
        \right).
    \end{multline*}
    As the factor inside the parentheses is nonpositive and we
    assume $x \geq 0$ and $y \geq 1$, we see that
    $\ln f(x,y)$, and hence $f(x,y)$, is nonincreasing in $x$.
    
    We desire an upper bound for $f(x,y)^y$, so we set $x = 0$:
    \[
        f(x,y)^y \leq f(0,y)^y
            = \left(\frac{y}{y-\frac{1}{2}}\right)^y
            = \left(\frac{1}{1 - \frac{1}{2y}}\right)^y
            \leq \left(1 + \frac{1}{y}\right)^y \leq e,
    \]
    where the penultimate inequality holds for $y \geq 1$.
\end{proof}

We can now prove \Cref{thm:angle_distr_bound}.

\begin{proof}[Proof of \Cref{thm:angle_distr_bound}]
    Observe that the upper bound on the density 
    of $\phi$ is independent of the orientation of the line $L$. Hence, we
    rotate $L$ about $x$ such that $L$ passes through $\mu$. We begin by
    proving the first part of the theorem.
    
    Let $f_Y$ denote the density of $y$,
    \[
        f_Y(y) = \frac{1}{(2\pi)^{d/2}\sigma^d}e^{-\frac{\|y-\mu\|^2}{2\sigma^2}}. 
    \]
    We center our coordinate system on $x$, and orient the $y_1$-axis
    along $\mu - x$, so that $\mu = (s, 0, \ldots, 0)$.
    We then switch to spherical coordinates
    $(r, \phi, \theta_1, \ldots, \theta_{d-2})$, where
    \begin{align*}
        y_1 &= r \cos\phi, \\
        y_2 &= r\sin\phi\cos\theta_1, \\
        y_3 &= r\sin\phi\sin\theta_2\cos\theta_2, \\
        &\vdots \\
        y_d &= r\sin\phi \sin\theta_2 \ldots \sin\theta_{d-3}\sin\theta_{d-2}.
    \end{align*}
    Here, $r$ ranges from $0$ to $\infty$, $\theta_{d-2}$ ranges from
    $0$ to $2\pi$, while all other angles range from $0$ to $\pi$. Due to
    the orientation of our coordinate system, the coordinate angle $\phi$
    corresponds to the random variable $\phi$ from the theorem statement.
    
    To compute the density of $\phi$ conditioned on
    $R = r$, we write
    \[
        f_{\phi|R=r}(\phi) = \frac{f_{\phi,R}(\phi, r)}{f_R(r)},
    \]
    where $f_{\phi,R}$ denotes the joint density of $\phi$ and
    $R$. We obtain this density by integrating the density of 
    $f_Y$ transformed to spherical coordinates over $\theta_1$ through
    $\theta_{d-2}$. Meanwhile, $f_R$ denotes the density of
    $R$, which is a noncentral $d$-dimensional chi distributed
    random variable with parameter $s$.
    
    The joint density $f_{\phi,R}$ is
    \begin{align*}
        f_{\phi,R}(\phi, r) = \frac{1}{(2\pi)^{d/2}} \frac{r^{d-1}}{\sigma^d}
            e^{-\frac{r^2 + \sigma^2}{2\sigma^2}} e^{\frac{rs\cos\phi}{\sigma^2}}
               \sin^{d-2}\phi\int_0^{2\pi}\dd \theta \prod_{k=1}^{d-3}\int_{0}^\pi \sin^k \theta \dd \theta.
    \end{align*}
    It holds that, for $k \in \mathbb{N}$,
    \[
        \int_0^\pi \sin^k \theta \dd \theta = \frac{\sqrt{\pi}\Gamma\left(\frac{k+1}{2}\right)}
            {\Gamma\left(\frac{k+2}{2}\right)}.
    \]
    By telescoping, it follows that
    \[
        \prod_{k=1}^{d-3} \int_0^\pi \sin^k \theta\dd \theta
            = \pi^{\frac{d-3}{2}} \cdot \frac{\Gamma(1)}{\Gamma\left(\frac{d-1}{2}\right)}
            = \frac{\pi^{\frac{d-3}{2}}}{\Gamma\left(\frac{d-1}{2}\right)}.
    \]
    Inserting this into our expression for $f_{\phi,R}$,
    we obtain
    \[
         f_{\phi,R}(\phi, r)
            \leq \frac{2^{1-\frac{d}{2}}}{\sqrt{\pi}}
            \frac{r^{d-1}}{\sigma^d}\frac{\sin^{d-2}\phi}{\Gamma\left(\frac{d-1}{2}\right)}
            e^{-\frac{r^2+s^2}{2\sigma^2}}e^{\frac{rs\cos\phi}{\sigma^2}}.
    \]
    
    Next, we use the expression for $f_R$ given in \Cref{eq:noncentral_chi}.
    Combining this with the above bound for $f_{\phi,R}$, we have
    \[
        f_{\phi|R=r}(\phi) \leq \frac{2^{1-\frac{d}{2}}}{\sqrt{\pi}}
            \frac{\sin^{d-2}\phi}{\Gamma\left(\frac{d-1}{2}\right)}
        \left(\frac{rs}{\sigma^2}\right)^{\frac{d}{2}-1} \frac{e^{\frac{rs\cos\phi}{\sigma^2}}}
            {I_{d/2-1}(rs/\sigma^2)}.
    \]
    
    For brevity, let
    $x := rs/\sigma^2$, and let $\nu := d/2 - 1$.
    Then, up to a constant, $f_{\phi|R=r}$ is bounded from above by 
    \begin{align}\label{eq:conditional_distr}
        \frac{x^\nu\sin^{2\nu}(\phi) e^{x\cos\phi}}
        {2^\nu \Gamma\left(\nu + \frac{1}{2}\right)I_\nu(x)}.
    \end{align}

    For any fixed $x$ and $\nu$, \Cref{eq:conditional_distr} is maximized when
    $\phi = \phi^*$, where $\phi^*$ satisfies
    \begin{align}\label{eq:optimal_angle}
        \sin^2 \phi^* = \frac{2\nu}{x} \cos\phi^*. 
    \end{align}
    Obtaining this is a matter of ordinary calculus.
    This equation has a unique solution in $[0, \pi]$ of
    \begin{align}\label{eq:optimal_angle2}
        \phi^* = 2\arctan\left(
            \sqrt{\frac{\sqrt{\nu^2 + x^2} - x}{\nu}} 
        \right) = 
             2\arctan\left(
            \sqrt{\sqrt{x^2/\nu^2 + 1} - x/\nu}
            \right).
    \end{align}
    It can also be verified that
    \begin{align}\label{eq:cosphi_opt}
        \cos\phi^* = \sqrt{1 + \frac{\nu^2}{x^2}} - \frac{\nu}{x} = \frac{\sqrt{x^2 + \nu^2} - \nu}{x}. 
    \end{align}
    Using this identity together with \Cref{eq:optimal_angle} in
    \Cref{eq:conditional_distr}, we find
    \begin{align*}
        f_{\phi|R=r}(\phi) &\leq
            \Theta(1) \cdot \frac{\nu^\nu}{\Gamma\left(\nu +\frac{1}{2}\right)}
                \cdot \left(\frac{\sqrt{x^2 + \nu^2} - \nu}{x}\right)^\nu
                \cdot \frac{e^{x\left(\sqrt{1 + \frac{\nu^2}{x^2}} - \frac{\nu}{x}\right)}}{I_\nu(x)} \\
                &= \Theta(1) \cdot \frac{(\nu/e)^\nu}{\Gamma\left(\nu + \frac{1}{2}\right)}
                \cdot \left(\frac{\sqrt{x^2 + \nu^2} - \nu}{x}\right)^\nu
                    \cdot \frac{e^{x\left(\sqrt{1 + \frac{\nu^2}{x^2}}\right)}}{I_\nu(x)} \\
                &= \Theta(1)
                    \cdot \left(\frac{\sqrt{x^2 + \nu^2} - \nu}{x}\right)^\nu
                    \cdot \frac{e^{\sqrt{x^2 + \nu^2}}}{I_\nu(x)},
    \end{align*}
    since Stirling's Formula yields $(\nu/e)^\nu/\Gamma(\nu + 1/2) = \Theta(1)$.
    
    We consider two cases, $x \leq 1$ and $x > 1$.
    ~\paragraph*{Case 1: $x \leq 1$.}
    We apply \Cref{lemma:besselbound} to
    \Cref{eq:conditional_distr},
    and find an upper bound of
    \[
        O\left(\frac{\Gamma(\nu + 1)}{\Gamma(\nu + 1/2)}\right) = O(\sqrt{\nu}).
    \]
    ~\paragraph*{Case 2: $x > 1$.}
    We use \Cref{lemma:besselbound_generic}, which yields
    \begin{align*}
        f_{\phi|R=r}(\phi) &\leq \Theta(1) \cdot
            \sqrt{\frac{x}{\sqrt{x^2 + \nu^2} - \nu}} \cdot \sqrt{x}
            = \Theta(1) \cdot \sqrt{\frac{\sqrt{x^2 + \nu^2} + \nu}{x}}\cdot \sqrt{x} \\
            &= O\left(\sqrt{x} + \sqrt{\nu}\right).
    \end{align*}
    Inserting the definitions of $x$ and $\nu$ concludes the proof of the first part.
    
    Next, let $d \geq 3$, or equivalently, $\nu \geq \frac{1}{2}$. 
    We assume $x > 1$ in the following; the case $x \leq 1$ simply follows
    from using \Cref{lemma:besselbound} in \Cref{eq:conditional_distr} and
    dividing by $\sin\phi$.
    
    To bound
    $f_{\phi|R=r}(\phi)/\sin\phi$, we follow mostly the same process. We return
    once more to \Cref{eq:conditional_distr}, and divide by $\sin\phi$.
    For any fixed $x$ and $\nu$, the resulting equation is then maximized when $\phi=\phi^*$,
    where $\phi^*$ satisfies
    \[
        \sin^2\phi^* = \frac{2\nu-1}{x}\cos\phi^*.
    \]
    The angle $\phi^*$ satisfies \Cref{eq:optimal_angle2,eq:cosphi_opt}, with
    $\nu$ replaced by $\nu - \frac{1}{2}$. Inserting this in \Cref{eq:conditional_distr} and
    working through the algebra, we eventually obtain
    \begin{multline*}
        \frac{f_{\phi|R=r}(\phi)}{\sin\phi} \leq \Theta(1) \cdot
            \frac{\left(\frac{\nu - \frac{1}{2}}{e}\right)^{\nu-\frac{1}{2}}}{\Gamma\left(\nu + \frac{1}{2}\right)}
                \cdot \sqrt{x} \cdot \sqrt{\frac{\sqrt{x^2 + \left(\nu - \frac{1}{2}\right)^2} + \nu - \frac{1}{2}}{x}}
                     \\ \cdot \left(
                        \frac{\sqrt{x^2 + \left(\nu - \frac{1}{2}\right)^2} - \left(\nu - \frac{1}{2}\right)}{x}
                    \right)^\nu
                    \cdot \frac{\exp\left(\sqrt{x^2 + \left(\nu-\frac{1}{2}\right)^2}\right)}{I_\nu(x)}.
    \end{multline*}
    Observe that for $\nu \geq \frac{1}{2}$, we have
    \[
        \frac{\left(\frac{\nu - \frac{1}{2}}{e}\right)^{\nu - \frac{1}{2}}}{\Gamma\left(\nu + \frac{1}{2}\right)}
            \leq 
        \frac{(\nu/e)^{\nu - \frac{1}{2}}}{\Gamma\left(\nu + \frac{1}{2}\right)}
            =
        \frac{(\nu/e)^{\nu}}{\Gamma\left(\nu + \frac{1}{2}\right)} \cdot \sqrt{\frac{e}{\nu}}
            \in O\left(\frac{1}{\sqrt{\nu}}\right).
    \]

    Since we assume $x > 1$, we may apply \Cref{lemma:besselbound_generic} to find
    \begin{multline*}
        \frac{f_{\phi|R=r}(\phi)}{\sin\phi} \leq
        \Theta(1)
            \cdot 
                \sqrt{\frac{x}{\nu}} \cdot 
                    \sqrt{\frac{\sqrt{x^2 + \left(\nu - \frac{1}{2}\right)^2} + \nu - \frac{1}{2}}{x}}
                        \left(\frac{\sqrt{x^2 + \left(\nu - \frac{1}{2}\right)^2} - \left(\nu - \frac{1}{2}\right)}{x}
                    \right)^\nu
                \\ \cdot \sqrt{x} \cdot \left(
                    \frac{x}{\sqrt{x^2 + \nu^2} - \nu} 
                \right)^{\nu + \frac{1}{2}}.
    \end{multline*}
    Through some more elementary algebra, we can bound this (up to a constant) by
    \[
        \frac{\sqrt{x^2 + \nu^2} + \nu}{\sqrt{\nu}} 
            \cdot \left(\frac{\sqrt{x^2 + \nu^2} + \nu}{\sqrt{x^2 + \left(\nu - \frac{1}{2}\right)^2} + \nu - \frac{1}{2}}\right)^\nu.
    \]
    The first factor in this expression evaluates to $O(\sqrt{\nu} + x/\sqrt{\nu})$. To conclude,
    we must show that
    \[
            \left(
                \frac{\sqrt{x^2 + \nu^2} + \nu}{\sqrt{x^2 + \left(\nu-\frac{1}{2}\right)^2}
                    + \left(\nu - \frac{1}{2}\right)}
            \right)^\nu \in O(1)
    \]
    for $\nu \in \{1/2, 1, 3/2, \ldots\}$. For $\nu = \frac{1}{2}$, we have
    \[
        \left(\frac{\sqrt{x^2 + \frac{1}{4}} + \frac{1}{2}}{x}\right)^{\frac{1}{2}}
            \leq \sqrt{1 + \frac{1}{x}} < \sqrt{2},
    \]
    where the latter inequality holds for $x > 1$. For $\nu \geq 1$, we use
    \Cref{lemma:simple_inequality} to bound the given quantity by $e$.
    This then proves the second part of the theorem.
\end{proof}


\section{Analysis of Single 2-Changes}\label{sec:single_2change}

To improve upon the previous analyses,
it pays to examine where the analysis
of Euclidean 2-opt with Gaussian perturbations \cite{mantheySmoothedAnalysis2Opt2013}
fails for $d \in \{2, 3\}$. The problem is that in the course of the proof,
Manthey \& Veenstra compute
\[
    \int_0^\infty \frac{1}{x^2}\chi_{d-1}(x)\dd{x},
\]
where $\chi_d$ denotes the $d$-dimensional chi distribution. This integral is
finite only when $d \geq 4$.

This problem does not appear in the results obtained by Englert et al.\ \cite{englertWorstCaseProbabilistic2014}.
They consider a more general model of smoothed analysis
wherein the adversary specifies a probability density for each point in the
TSP instance independently. Since the only information available on the probability
densities is their upper bound, they consider a simplified model of a 2-change to
keep the analysis tractable.
The analysis is then translated to their generic model, which incurs a factor which
is super-exponential in $d$.

Even when one considers $d$ to be a constant as Englert et al.\ do, the genericity of
their model still comes at a cost when translated to a smoothed analysis with
Gaussian perturbations, eventually yielding a bound which is
polynomial in $\sigma^{-d}$.

Specifying the perturbations as Gaussian enables us to analyze the true
random experiment modeling a 2-change more closely, as we know the distributions 
of the distances between points in the smoothed instance. Combined with
\Cref{thm:angle_distr_bound}, which provides information on the angles between
edges in the instance, we can carry out an analysis that improves on both
{Englert et al.}'s as well as Manthey \& Veenstra's result when we consider
Gaussian perturbations.

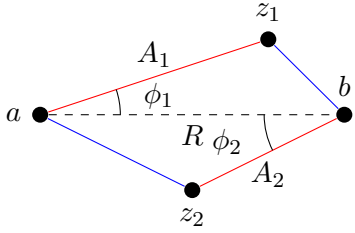
\begin{figure}
\centering
\begin{tikzpicture}
    
    \vertex[label=left:$a$](a) at (0, 0) {};
    \vertex[label=above:$z_1$](z1) at (3, 1) {};
    \vertex[label=above:$b$](b) at (4, 0) {};
    \vertex[label=below:$z_2$](z2) at (2, -1) {};
    
    \draw[red] (a) -- node[above, black] {$A_1$} ++ (z1);
    \draw[blue] (z1) -- (b); 
    \draw[blue] (z2) -- (a);
    \draw[red] (b) -- node[below, black] {$A_2$} ++ (z2);
    
    \draw[dashed] (a) -- node[below] {$R$} ++ (b);
    
    \pic [draw, -, "$\phi_1$", angle eccentricity=1.5, angle radius=30] {angle = b--a--z1};
    \pic [draw, -, "$\phi_2$", angle eccentricity=1.5, angle radius=30] {angle = a--b--z2};
    
\end{tikzpicture}
\caption{Labels of points and angles involved in a single 2-change.
\label{fig:2change}}
\end{figure}

We first set up our model of a 2-change perturbed by Gaussian random variables.
To obtain a bound for this case, we first formulate a different analysis of single
2-changes. Consider a 2-change involving the points
$\{a, b, z_1, z_2\} \subseteq [-D, D]^d$,
where the edges
$\{a, z_1\}$ and $\{b, z_2\}$ are replaced by $\{b, z_1\}$ and $\{a, z_2\}$.
The improvement to the tour length due to this 2-change is
\[
    \Delta = \|a - z_1\| - \|b - z_1\| + \|b - z_2\| - \|a - z_2\|.
\]

To analyze $\Delta$, we first define
$A_1 := \|a - z_1\|$, $A_2 := \|b - z_2\|$, and $R := \|a - b\|$.
Moreover, 
we identify the angle $\phi_1$ as the angle between $a - z_1$ and $a - b$, and
restrict it to $[0, \pi]$. The corresponding angle $\phi_2$ is defined similarly.
The restriction of these angles to $[0, \pi]$ is without loss of generality;
one may readily observe from
\Cref{fig:2change} that flipping the sign of either $\phi_1$ or $\phi_2$ does not change
the value of $\Delta$.

While \Cref{fig:2change} may give the impression that we are restricting the analysis
to the $d = 2$ case, the analysis is valid for any $d \geq 2$. The two triangles
$\triangle a z_1 b$ and $\triangle a z_2 b$ will lie in two separate planes in general.
The distances involved must thus be understood as $d$-dimensional Euclidean distances.

With these definitions, we have $\Delta = \eta_1 + \eta_2$, where for $i \in [2]$
\[
    \eta_i = A_i - \sqrt{A_i^2 + R^2 - 2A_i R \cos\phi_i},
\]
which follows from the Law of Cosines.

Suppose we condition on the events $A_1 = a_1$, $A_2 = a_2$, and
$R = r$, for some $a_1, a_2, r > 0$. Under these events,
$\eta_1$ and $\eta_2$ are independent random variables.
Moreover, $\Delta$ is completely fixed by
revealing the angles $\phi_1$ and $\phi_2$.
Since we condition on $A_i = a_i$ and $R = r$, we can
then bound the density of $\phi_i$
using \Cref{cor:angle_distr_bound}.

We can use this independence to obtain bounds for
$\prob(\Delta \in (0, \epsilon])$ for some
small $\epsilon > 0$ under these events,
for various orderings of $a_1$, $a_2$ and $r$. These
bounds are given in
\Cref{lemma:single_2change_conditioned}.

We begin by obtaining a bound to the density of $\eta_i$, $i \in [2]$, using the 
fact that all randomness in $\eta_i$ is contained in the angle 
$\phi_i$ under the conditioning that $A_i = a_i$ and $R = r$. We denote by
$f_{\phi_i|R=r,A_i=a_i}$ the density of the angle $\phi_i$, conditioned
on $R = r$ and $A_i = a_i$.

\begin{lemma}\label{lemma:eta_density}
    Let $i \in [2]$.
    The density of $\eta_i = \|a - z_i\| - \|b - z_i\|$, conditioned
    on $A_i = a_i$ and $R = r$, is bounded from above
    by
    \[
        \frac{a_i + r}{a_i r} \cdot \frac{f_{\phi_i|R=r,A_i=a_i}(\phi_i(\eta))}{|\sin\phi_i(\eta)|},
    \]
    where
    $
        \phi_i(\eta) = \arccos\left(
            \frac{a_i^2 + r^2 - (a_i - \eta)^2}{2a_ir} 
        \right).
    $
\end{lemma}

\begin{proof}
    Let the conditional density of $\eta_i$ be
    $f_{\eta_i|R=r,A_i=a_i}$. Since $\phi_i$ is restricted to $[0, \pi]$
    by assumption, there exists a bijection between $\eta_i$ and $\phi_i$. To
    be precise, we have
    \[
        \phi_i(\eta_i) = \arccos\left(
            \frac{a_i^2 + r^2 - (a_i - \eta_i)^2}{2a_ir} 
        \right).
    \]
    By standard transformation rules of probability densities, it holds that
    \[
        f_{\eta_i|R=r,A=a_i}(\eta) = \left|\der[\phi_i(\eta)]{\eta}\right|
            f_{\phi_i|R=r,A_i=a_i}(\phi_i(\eta)).
    \]
    The derivative is easily evaluated:
    \[
        \der[\phi_i(\eta)]{\eta} = \frac{-1}{\sqrt{1 - \left(\frac{a_i^2 + r^2 - (a_i-\eta)^2}{2a_i r}\right)}} \cdot
            \frac{a_i - \eta}{a_i r}
            = \frac{-1}{\sin\phi(\eta)} \cdot \frac{a_i-\eta}{a_i r}.
    \]
    
    Finally, we have
    $a_i - \eta \leq a_i + r$, which follows from the triangle
    inequality.
    This concludes the proof.
\end{proof}

With \Cref{cor:angle_distr_bound}, we have an upper bound for
$f_{\phi_i|R=r,A_i = a_i}$.
Unfortunately, simply inserting this upper bound
is not enough for us to bound $f_{\eta_i|A_i=a_i,R=r}$, since the density 
as obtained from \Cref{lemma:eta_density} diverges for $\phi = 0$
and $\phi = \pi$. There is however a way to cure this divergence.

We now consider a full 2-change (cf.\ \Cref{fig:2change}).
To analyze the improvement $\Delta$ caused by this 2-change,
we construct a random experiment, conditioned on the outcomes
$A_1 = a_1$, $A_2 = a_2$, and $R = r$. We write this random experiment
in Algorithm \ref{alg:randomexpt}, since we will need to execute different
experiments depending on the ordering of the values of
$a_1$, $a_2$ and $r$. The parameters $b_1$ and $b_2$ of this algorithm
will take values in $\{a_1, a_2, r\}$, depending on this ordering.

\begin{algorithm}
\caption{The algorithm we use to model a random 2-change with fixed $A_1 = a_1$, $A_2 = a_2$, and $R = r$.\label{alg:randomexpt}}
\begin{algorithmic}[1]
    \Function{RandomExpt}{$b_1$, $b_2$}
        \State Draw $\phi_1 \sim f_{\phi|R=r,A_1 = a_1}$
        \State Draw $\phi_2 \sim f_{\phi|R=r,A_2 = a_2}$
        \If{$\sqrt{b_1}\sin \phi_1 > \sqrt{b_2}\sin \phi_2$}
            \State \Return $(1, \phi_1)$
        \Else
            \State \Return $(2, \phi_2)$
        \EndIf
    \EndFunction
\end{algorithmic}
\end{algorithm}

The function \texttt{RandomExpt} outlined in Algorithm \ref{alg:randomexpt} branches
on the outcome of the variable
$Z_i = \sqrt{b_i}\sin \phi_i$, $i \in [2]$, where $b_i$ is some distance;
we will choose $b_i$ among $\{r, a_i\}$ in subsequent lemmas.

Note that \verb|RandomExpt| returns a tuple $(i, \phi)$, where
$i \in [2]$. We call the angle returned by \texttt{RandomExpt}
the
\emph{good angle}. Moreover, we label the event
$i = 1$ as $E_1$, and $i = 2$ by $E_2$.
The crux of the analysis is now to analyze $\eta_1$ if $E_1$ occurs,
and $\eta_2$ if $E_2$ occurs, as under $E_i$ the density of $\eta_i$ is
bounded from above.

\begin{lemma}\label{lemma:goodangle_density}
    Let $(i, \phi) = \texttt{RandomExpt}(b_1, b_2)$ for some
    $b_1, b_2 > 0$.
    Let $j = 3 - i$.
    The density of $\phi$, conditioned on $R = r$, $A_1 = a_1$, $A_2 = a_2$,
    is then bounded from above by
    \[ 
        \frac{2M_{\phi_1}M_{\phi_2}}{\prob(E_i)} \cdot
            \arcsin\left(\min\left\{1, \sqrt{\frac{b_i}{b_j}}\sin\phi\right\}\right),
    \]
    where $M_{\phi_i} = \max_{0 \leq \phi \leq \pi} f_{\phi_i|R=r,A_i=a_i}(\phi)$.
\end{lemma}

\begin{proof}
    We omit the conditioning on $A_1 = a_1$,
    $A_2 = a_2$ and $R = r$ in the following, for the sake of clarity.
    We prove only the case $i = 1$, thus conditioning on $E_1$, as the proof for $i = 2$
    proceeds essentially identically.
    
    Let $X_i = \sqrt{b_i}\sin\phi_i$, $i \in [2]$. The event $E_1$ is then
    equivalent to $X_1 > X_2$. Let $Z$ in turn denote the random variable
    given by $X_1$ conditioned on $E_1$. The cumulative distribution
    function of $Z$ is equal to
    \[
        F_Z(x) = \prob(X_1 \leq x \given X_1 > X_2)
            = \frac{\prob(X_1 \leq x \wedge X_1 > X_2)}{\prob(E_1)}.
    \]
    By the independence of $X_1$ and $X_2$, this is equal to
    \[
        F_Z(x) = \frac{1}{\prob(E_1)} \cdot \int_0^x f_{X_1}(y) \int_0^y f_{X_2}(z)\dd z \dd y.
    \]
    Computing the density of $Z$ is then simply a matter of differentiation. Since
    $\prob(E_1)$ does not depend on $x$, we obtain
    \[
        f_Z(x) = \frac{1}{\prob(E_1)} \cdot f_{X_1}(x)\int_0^x f_{X_2}(z)\dd z. 
    \]
    
    We next require the density of $X_i = \sqrt{b_i}\sin\phi_i$. Observe that
    \begin{align}\label{eq:phitox}
        \prob(X_i \leq x) = 
            \prob\left(\phi_i \leq  \arcsin(x/\sqrt{b_i})\right)
            + \prob\left(\phi_i \geq  \pi - \arcsin(x/\sqrt{b_i})\right).
    \end{align}
    Differentiating this expression to $x$, we find for
    $x < \sqrt{b_i}$
    \begin{align*}
        f_{X_i}(x) &= \der{x} \left(
            \prob\left(\phi_i \leq \arcsin(x/\sqrt{b_i})\right)
            + 1 - \prob\left(\phi_i \geq \pi - \arcsin(x/\sqrt{b_i})\right)
        \right) \\
            &= 
            \der{x}\left(\arcsin\left(\frac{x}{\sqrt{b_i}}\right)\right) \cdot
            \left[f_{\phi_i}\left(\arcsin\left(\frac{x}{\sqrt{b_i}}\right)\right) 
        + f_{\phi_i}\left(\pi - \arcsin\left(\frac{x}{\sqrt{b_i}}\right)\right)\right] \\
        &= \frac{1}{\sqrt{b_i - x^2}} \cdot \left[
        f_{\phi_i}\left(\arcsin\left(\frac{x}{\sqrt{b_i}}\right)\right) 
        + f_{\phi_i}\left(\pi - \arcsin\left(\frac{x}{\sqrt{b_i}}\right)\right)
        \right],
    \end{align*}
    and $0$ for $x \geq \sqrt{b_i}$. Letting
    $M_{\phi_i} = \max_{0 \leq \phi \leq \pi} f_{\phi_i|R=r,A_i=a_i}(\phi)$,
    which exists by \Cref{cor:angle_distr_bound}, we obtain
    \[
        f_{X_i}(x) \leq 2M_{\phi_i}\cdot \begin{cases}
            \frac{1}{\sqrt{b_i - x^2}}, & \text{if } x < \sqrt{b_i}, \\
            0, & \text{otherwise.}
        \end{cases}
    \]
    
    Using this density, together with the identity
    $\int_0^x (\sqrt{b} - y^2)^{-1/2}\dd y = \arcsin(x/\sqrt{b})$
    for $x < \sqrt{b}$, we obtain
    \[
        f_{Z}(x) \leq \frac{2M_{\phi_1}M_{\phi_2}}{\prob(E_1)} \cdot \frac{\arcsin\left(\min\left\{1, \frac{x}{\sqrt{b_2}}\right\}\right)}
            {\sqrt{b_1 - x^2}}
    \]
   
    if $x < \sqrt{b_1}$, and $f_Z(x) = 0$ otherwise.
    It remains to convert $Z$ back to $\phi$, where $\phi$ is the good angle.
    Since we have conditioned on $E_1$, we know that
    $Z = \sqrt{b_1}\sin\phi$. Using similar considerations as used
    in \Cref{eq:phitox}, we have
    \begin{align*}
        f_Z(x) = 
            \frac{1}{\sqrt{b_1 - x^2}}f_{\phi}(\arcsin(x/\sqrt{b_1}))
            +\frac{1}{\sqrt{b_1 - x^2}}f_{\phi}(\pi - \arcsin(x/\sqrt{b_1})).
    \end{align*}
    Since this expression holds for all $x \in (0, \sqrt{b_1})$, and since
    probability densities are non-negative, it follows
    that
    \[
        f_{\phi}(\phi) \leq \frac{2M_{\phi_1}M_{\phi_2}}{\prob(E_1)} \cdot
            \arcsin\left(\min\left\{1, \sqrt{\frac{b_1}{b_2}}\sin\phi\right\}\right),
    \]
    for all $\phi \in (0, \pi)$. 
\end{proof}

For the next part, we apply \Cref{lemma:goodangle_density} to
\Cref{lemma:eta_density} to bound
the density of $\eta_i$, given that $E_i$ occurs.

\begin{lemma}\label{lemma:goodeta_density}
    Let $i \in [2]$ and
    $j = 3 - i$.
    Let $f_{\eta_i|E_i}$ denote the density
    of $\eta_i$, conditioned on
    $E_i$ as well as the outcomes $R = r$, $A_1 = a_1$, and $A_2 = a_2$. Then
    \begin{align*}
            f_{\eta_i|E_i}(\eta) \leq \frac{1}{\prob(E_i)} \cdot
                \frac{2\pi M_{\phi_1} M_{\phi_2}}{\min\{a_1, r\}\min\{a_2, r\}},
    \end{align*}
    where $M_{\phi_i} = \max_{0 \leq \phi \leq \pi} f_{\phi_i|R=r,A_i=a_i}(\phi)$.
\end{lemma}

\begin{proof}
    We prove only the case $i = 1$.
    From
    \Cref{lemma:eta_density}, we know that
    \[
        f_{\eta_i|E_i}(\eta) \leq \frac{a_i + r}{a_i r} \cdot \frac{f_{\phi_i|E_i,A_1=a_1,A_2=a_2}(\phi)}
            {\sin\phi}.
    \]
    
    Let $(i, \phi) = \verb|RandomExpt|(b_1, b_2)$, for
    some $b_1, b_2 > 0$. We will choose values for
    $b_1$ and $b_2$ depending on the ordering of $a_1, a_2$
    and $r$. Note that we may do this, since
    we know the choices of $a_1$, $a_2$ and $r$ before executing
    \verb|RandomExpt|.
    
    Since we condition on $E_1$, we know that $i = 1$, and hence
    that $\phi_1$ is the good angle.
    By
    \Cref{lemma:goodangle_density}, we can obtain a bound for 
    $f_{\phi|E_i,A_1=a_1,A_2=a_2,R=r}$. We thus find
    \[
        f_{\eta_1|E_1}(\eta) \leq \frac{2M_{\phi_1}M_{\phi_2}}{\prob(E_1)}
        \cdot \frac{a_1 + r}{a_1 r} \cdot \frac{\arcsin\left(\min\left\{1, \sqrt{\frac{b_1}{b_2}}\sin\phi\right\}\right)}
            {\sin\phi}.
    \]
    First, suppose $\sin\phi \geq \sqrt{b_2/b_1}$. Then the arcsine
    evaluates to $\pi/2$, and so the above is bounded from above by
    \[
        \frac{\pi}{2}\sqrt{\frac{b_1}{b_2}}.
    \]
    Second, suppose $\sin\phi < \sqrt{b_2/b_1}$. Since $\arcsin(x) \leq \pi x/2$ for
    $x \in (0, 1)$, this case yields the same bound, and we obtain
    \[
        f_{\eta_1|E_1}(\eta) \leq \frac{\pi M_{\phi_1}M_{\phi_2}}{\prob(E_1)} 
        \cdot \frac{a_1 + r}{a_1 r} \cdot \sqrt{\frac{b_1}{b_2}}
    \]
    
    We now examine the four relevant orderings of $a_1$, $a_2$ and $r$.
    ~\paragraph*{Case 1: $a_1, a_2 \leq r$.} We let $b_1 = a_1$ and $b_2 = a_2$. Then we
    have
    \[
        \frac{a_1 + r}{a_1 r} \cdot \sqrt{\frac{a_1}{a_2}}
            = \frac{a_1 + r}{r\sqrt{a_1 a_2}}
            \leq \frac{2r}{r\sqrt{a_1 a_2}}
            = \frac{2}{\sqrt{a_1 a_2}}.
    \]
    ~\paragraph*{Case 2: $a_1, a_2 \geq r$.} We let $b_1 = b_2 = r$, and obtain
    \[
        \frac{a_1 + r}{a_1 r} \leq \frac{2a_1}{a_1 r} = \frac{2}{r}.
    \]
    ~\paragraph*{Case 3: $a_1 \geq r \geq a_2$.} We let $b_1 = r$ and $b_2 = a_2$, which
    yields
    \[
        \frac{a_1 + r}{a_1 r} \cdot \sqrt{\frac{r}{a_2}}
            = \frac{a_1 + r}{\sqrt{a_2 r} a_1} \leq \frac{2}{\sqrt{a_2 r}}.
    \]
    ~\paragraph*{Case 4: $a_2 \geq r \geq a_1$.} We let $b_1 = a_1$ and $b_2 = r$, to find
    \[
        \frac{a_1 + r}{a_1 r} \sqrt{\frac{a_1}{r}} \leq \frac{2r\sqrt{a_1}}{a_1 r\sqrt{r}}
        = \frac{2}{\sqrt{a_1 r}}.
    \]
    This final case concludes the proof.
\end{proof}

The bound on the density of $\eta_i$ from \Cref{lemma:goodeta_density}
puts us in the position to prove a bound on the
probability that $\Delta \in (0,\epsilon]$.

\begin{lemma}\label{lemma:single_2change_conditioned}
    Let $\Delta$ denote the improvement of a 2-change.
    Then
     \begin{align*}
        \prob(\Delta \in (0, \epsilon] \given A_1 = a_1, A_2 = a_2, R = r) \leq 
            \frac{\pi M_{\phi_1} M_{\phi_2}\epsilon}{\min\{a_1, r\}\min\{a_2, r\}},
    \end{align*}
    where $M_{\phi_i} = \max_{0 \leq \phi \leq \pi} f_{\phi_i|R=r,A_i=a_i}(\phi)$.
\end{lemma}

\begin{proof}
    We condition first on $E_1$, and then let an adversary
    choose an outcome for $\eta_2$, say, $\eta_2 = t$. Then we have
    $\Delta \in (0, \epsilon]$ iff $\eta_1 \in (-t, -t+\epsilon]$,
    which is an interval of size $\epsilon$. 
    
    Since the probability that $\eta_1$ falls into an interval of size $\epsilon$
    is at most $\epsilon \cdot \max_{\eta} f_{\eta_1|E_1}(\eta)$, all we need to
    conclude the proof for $E_1$ is a bound on $f_{\eta_1|E_1}(\eta)$. This is provided
    by \Cref{lemma:goodeta_density}. 
    
    We then repeat the same argument for $E_2$. The result is obtained by
    applying the Law of Total Probability.
\end{proof}

With \Cref{lemma:single_2change_conditioned}, we could prove a bound
on the smoothed complexity of 2-opt already. However, the resulting bound
would be weaker than existing results. Instead of analyzing single
2-changes, we thus use the framework of linked pairs of 2-changes
in \Cref{sec:linked_pairs}.

For the analysis in \Cref{sec:linked_pairs}, it is convenient to have
some lemmas similar to \Cref{lemma:single_2change_conditioned},
with one or more of the distances $A_1$, $A_2$ and $R$ integrated out.
These are given in
\Cref{lemma:single_2change_one_distance1,lemma:single_2change_one_distance3,lemma:single_2change_one_distance_5}.
The proofs are straightforward computations.

\begin{lemma}\label{lemma:single_2change_one_distance1}
    For $i \in [2]$,
    \[
        \prob(\Delta \in (0, \epsilon] \given A_i = a_i, R = r)
            = O\left(
                \left(
                    \frac{\sqrt{d}D}{\sigma^2} + \frac{d}{\sqrt{a_i r}}
                        + \frac{d}{r} + \frac{d^{3/4}\sqrt{D}}{\sigma} \left(
                            \frac{1}{\sqrt{a_i}} + \frac{1}{\sqrt{r}}
                        \right)
                \right) \cdot \epsilon 
            \right).
    \]
\end{lemma}

\begin{proof}
    We assume $i = 1$, since by symmetry the result for $i = 2$
    follows essentially identically.
    
    Consider the cases $a_1 \leq r$ and $a_1 \geq r$ separately.
    ~\paragraph*{Case 1: $a_1 \leq r$.} For this case, we have by \Cref{lemma:single_2change_conditioned}
    for some constants $c, c', c'' > 0$,
    \begin{align*}
        \prob(\Delta \in (0, \epsilon]\given A_1 = a_1, A_2=a_2,R=r)
            &\leq c\cdot \frac{M_{\phi_1}M_{\phi_2}\epsilon}{\sqrt{a_1}} \cdot \begin{cases}
                    \frac{1}{\sqrt{r}}, & \text{if } a_2 \geq r \\
                    \frac{1}{\sqrt{a_2}}, & \text{if } a_2 \leq r
            \end{cases}\\
            &\leq c' \cdot \frac{M_{\phi_1}\epsilon}{\sqrt{a_1}} \cdot \begin{cases}
                \sqrt{\frac{d}{r}} + \frac{d^{1/4}\sqrt{D}}{\sigma}, & \text{if } a_2 \geq r \\
                \sqrt{\frac{d}{a_2}} + \frac{d^{1/4}\sqrt{D}}{\sigma}, & \text{if } a_2 \leq r
            \end{cases} \\
            &\leq c'' \cdot \frac{M_{\phi_1}\epsilon}{\sqrt{a_1}}
                \left(\sqrt{\frac{d}{r}} + \sqrt{\frac{d}{a_2}} + \frac{d^{1/4}\sqrt{D}}{\sigma}\right),
    \end{align*}
    where we use \Cref{cor:angle_distr_bound} to bound $M_{\phi_2}$.
    
    We can now use \Cref{lemma:chi_expect,lemma:chi_stochdom} to integrate out
    $a_2$, leaving us with
    \[
        O\left(
            \frac{M_{\phi_1}\epsilon}{\sqrt{a_1}} \left(
                \sqrt{\frac{d}{r}} + \frac{d^{1/4}}{\sqrt{\sigma}} + \frac{d^{1/4}\sqrt{D}}{\sigma}
            \right)
        \right).
    \]
    Using that $D \geq 1$ and $\sigma \leq 1$, we see that the third term in the inner
    brackets is at least as large as the second term, and so we obtain
    \[
        \prob(\Delta \in (0, \epsilon]\given A_1 = a_1, A_2=a_2,R=r) = O\left(
            \frac{M_{\phi_1}}{\sqrt{a_1}}\left(\sqrt{\frac{d}{r}} + \frac{d^{1/4}\sqrt{D}}{\sigma}\right) \cdot \epsilon
        \right).
    \]
    
    Now we use \Cref{cor:angle_distr_bound} to conclude
    $M_{\phi_1} = O\left(d^{1/4}\sqrt{Da_1}/\sigma\right)$, yielding
    \begin{multline*}
        O\left(
            \left(\sqrt{\frac{d}{a_1}} + \frac{d^{1/4}\sqrt{D}}{\sigma}\right)
                \cdot
            \left(\sqrt{\frac{d}{r}} + \frac{d^{1/4}\sqrt{D}}{\sigma}\right) \cdot \epsilon
        \right) \\ 
        =
        O\left(
            \left( 
                \frac{d}{\sqrt{a_1 r}} + \frac{d^{3/4}\sqrt{D}}{\sigma}\left(
                \frac{1}{\sqrt{a_1}} + \frac{1}{\sqrt{r}} 
            \right)
             + \frac{\sqrt{d}D}{\sigma^2}
             \right) \cdot \epsilon
        \right).
    \end{multline*}

    ~\paragraph*{Case 2: $a_1 \geq r$.} 
    Here, \Cref{lemma:single_2change_conditioned} tells us
    \begin{align*}
        \prob(\Delta \in (0, \epsilon]\given A_1 = a_1, A_2=a_2,R_1=r))
            &\leq c\cdot  M_{\phi_1}M_{\phi_2} \epsilon \cdot 
                \begin{cases}
                    \frac{1}{r}, & \text{if } a_2 \geq r \\
                    \frac{1}{\sqrt{r a_2}}, & \text{if } a_2 \leq r
                \end{cases} \\
            &\leq c' \cdot M_{\phi_1}\epsilon \cdot \begin{cases}
                \frac{\sqrt{d}}{r} + \frac{d^{1/4}\sqrt{D}}{\sigma\sqrt{r}}, & \text{if } a_2 \geq r, \\
                \sqrt{\frac{d}{r a_2}} + \frac{d^{1/4}\sqrt{D}}{\sigma\sqrt{r}}, & \text{if } a_2 \leq r
            \end{cases} \\
            &\leq c'' \cdot M_{\phi_1} \epsilon  \cdot
                \left(
                    \frac{\sqrt{d}}{r} + \sqrt{\frac{d}{ra_2}} + \frac{d^{1/4}\sqrt{D}}{\sigma \sqrt{r}}
                \right),
    \end{align*}
    again for some $c, c', c'' > 0$ and
    using \Cref{cor:angle_distr_bound} to bound $M_{\phi_2}$.
    
    Integrating out $a_2$ using \Cref{lemma:chi_expect,lemma:chi_stochdom}, we have
    \[
        O\left(M_{\phi_1}\epsilon \cdot \left(
            \frac{\sqrt{d}}{r} + \frac{d^{1/4}}{\sqrt{\sigma r}} + \frac{d^{1/4}\sqrt{D}}{\sigma \sqrt{r}}
        \right)
        \right) \subseteq
        O\left(M_{\phi_1}\epsilon \cdot \left(
            \frac{\sqrt{d}}{r} + \frac{d^{1/4}\sqrt{D}}{\sigma \sqrt{r}}
        \right)\right).
    \]
    Using \Cref{cor:angle_distr_bound} to insert
    $M_{\phi_1} = O\left(\sqrt{d} + d^{1/4}\sqrt{Dr}/\sigma\right)$, we find
    \[
        O\left(\left(
                \frac{\sqrt{d}}{r} + \frac{d^{1/4}\sqrt{D}}{\sigma \sqrt{r}}
            \right) \cdot \left(
                \sqrt{d} + \frac{d^{1/4}\sqrt{Dr}}{\sigma}
            \right) \cdot \epsilon
            \right)
        = O\left(
        \left(
            \frac{d}{r} + \frac{d^{3/4}\sqrt{D}}{\sigma\sqrt{r}}  +
                \frac{\sqrt{d}D}{\sigma^2}
            \right) \cdot \epsilon
        \right).
    \]

    The result follows from these two cases.
\end{proof}

\begin{lemma}\label{lemma:single_2change_one_distance3}
    For $i \in [2]$,
    \[
        \prob(\Delta \in (0,\epsilon] \given A_i = a_i)
            = O\left(
            \left(
                \frac{\sqrt{d}D}{\sigma^2} + \frac{d^{3/4}\sqrt{D}}{\sigma \sqrt{a_i}}
            \right) \cdot \epsilon
        \right).
    \]
\end{lemma}

\begin{proof}
    From \Cref{lemma:single_2change_one_distance1},
    we have
    \[
        \prob(\Delta \in (0, \epsilon] \given A_i = a_i, R = r)
            = O\left(
                \left(
                    \frac{\sqrt{d}D}{\sigma^2} + \frac{d}{\sqrt{a_i r}}
                        + \frac{d}{r} + \frac{d^{3/4}\sqrt{D}}{\sigma} \left(
                            \frac{1}{\sqrt{a_i}} + \frac{1}{\sqrt{r}}
                        \right)
                \right) \cdot \epsilon 
            \right).
    \]
    We can then apply \Cref{lemma:chi_expect,lemma:chi_stochdom} to integrate out
    $r$. This leaves 
    \[
        O\left(
        \left(
            \frac{\sqrt{d}D}{\sigma^2}  + \frac{d^{1/4}}{\sqrt{a_i\sigma}}
                + \frac{\sqrt{d}}{\sigma} + \frac{\sqrt{dD}}{\sigma^{3/2}} 
                + \frac{d^{3/4}\sqrt{D}}{\sigma \sqrt{a_i}}
            \right) \cdot \epsilon
        \right) \subseteq O\left(
            \left(
                \frac{\sqrt{d}D}{\sigma^2} + \frac{d^{3/4}\sqrt{D}}{\sigma \sqrt{a_i}}
            \right) \cdot \epsilon
        \right),
    \]
    as claimed.
\end{proof}

\begin{lemma}\label{lemma:single_2change_one_distance_5}
    \[
        \prob(\Delta \in (0, \epsilon] \given R = r) = O\left(
            \left( 
            \frac{\sqrt{d}D}{\sigma^2} + \frac{d}{r} + \frac{d^{3/4}\sqrt{D}}{\sigma \sqrt{r}}
            \right) \cdot \epsilon
        \right).
    \]
\end{lemma}

\begin{proof}
    The result follows from taking
    \Cref{lemma:single_2change_one_distance3} and integrating out $a_i$
    using \Cref{lemma:chi_expect,lemma:chi_stochdom}.
\end{proof}

\section{Linked Pairs of 2-Changes}\label{sec:linked_pairs}

To obtain bounds on the smoothed complexity of 2-opt, we consider so-called
linked pairs of 2-changes, introduced previously by Englert et al.\ \cite{englertWorstCaseProbabilistic2014}.
A pair of 2-changes is said to be linked if some edge
removed from the tour by one 2-change is added to
the tour by the other 2-change.

Such linked pairs have been considered in several previous works
\cite{englertWorstCaseProbabilistic2014, mantheySmoothedAnalysis2Opt2013}. In each case,
the distinction has been made between several types of linked pairs.
In our analysis, only two of these types are relevant, and so we
will describe only these types for the sake of brevity.

We consider 2-changes which share exactly one edge, and subdivide them into
pairs of type 0 and of type 1.
A generic 2-change removes the edges $\{z_1, z_2\}$ and
$\{z_3, z_6\}$ while adding $\{z_1, z_6\}$
and $\{z_2, z_3\}$. The other 2-change
removes $\{z_3, z_4\}$ and $\{z_5, z_6\}$
while adding $\{z_3, z_6\}$ and $\{z_4, z_5\}$.
Note that $\{z_3, z_6\}$ occurs in both 2-changes.
\begin{itemize}
    \item If $|\{z_1, \ldots, z_6\}| = 6$, then we say the
        linked pair is of type 0.
    \item If $|\{z_1, \ldots, z_6\}| = 5$, then we say the linked
        pair is of type 1.
\end{itemize}

Type 1 can itself be subdivided into two types, 1a
and 1b. We will detail this distinction in
\Cref{sec:type1}. 

Before moving on to analyzing linked pairs, we state
a useful lemma that justifies limiting the
discussion to just linked pairs of types 0 and 1.

\begin{lemma}[\mbox{\cite[Lemma 9]{englertWorstCaseProbabilistic2014}}]\label{lemma:number_of_pairs}
    In every sequence of $t$ consecutive 2-changes the number
    of disjoint pairs of 2-changes of type 0 or type 1
    is at least $\Omega(t) - O(n^2)$.
\end{lemma}

\subsection{Type 0}\label{sec:type0}

We begin with type 0, as this is by far the simplest linked pair.
For clarity, see \Cref{fig:2change_pairs} (left) for an illustration of a
type 0 linked pair. It should be noted that, while \Cref{fig:2change_pairs}
shows a specific configuration of vertices in two dimensions,
the results of this section hold generally; the analysis does not
depend on any point having a particular orientation
with respect to its neighbors. The same holds
for the results in \Cref{sec:type1}.

The improvement of a  type 0
linked pair is completely specified by a small number of random variables.
We require five distances between vertices,
$R_1 = \|z_1 - z_3\|$, $A_1 = \|z_3 - z_6\|$, $A_2 = \|z_1 - z_2\|$,
$R_2 = \|z_4 - z_6\|$ $A_3 = \|z_4 - z_5\|$. Additionally,
we need the following angles:
\begin{enumerate}
    \item $\phi_1$ between $z_2-z_1$ and $z_3-z_1$,
    \item $\phi_2$ between $z_1-z_3$ and $z_6-z_3$,
    \item $\phi_1'$ between $z_3-z_6$ and $z_4-z_6$,
    \item $\phi_3$ between $z_6-z_4$ and $z_5-z_4$.
\end{enumerate}

Note that, if we condition on $A_1 = a_1$, the events
$\Delta_1 \in (0, \epsilon]$ and $\Delta_2 \in (0, \epsilon]$ are independent.
We can then apply \Cref{lemma:single_2change_conditioned}, together with several applications of
\Cref{lemma:chi_expect}.

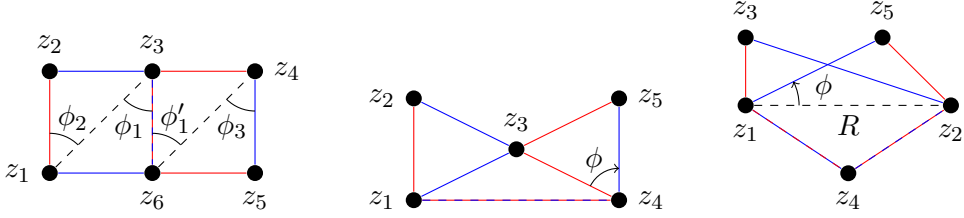
\begin{figure}
\centering
\begin{tikzpicture}[scale=0.9]
    
    \vertex[label=left:$z_1$](z1) at (0, 0) {};
    \vertex[label=above:$z_2$](z2) at (0, 1.5) {};
    \vertex[label=above:$z_3$](z3) at (1.5, 1.5) {};
    \vertex[label=right:$z_4$](z4) at (3, 1.5) {};
    \vertex[label=below:$z_5$](z5) at (3, 0) {};
    \vertex[label=below:$z_6$](z6) at (1.5, 0) {};
    
    \draw[red] (z1) -- (z2);
    \draw[blue] (z2) -- (z3); 
    \draw[red] (z3) -- (z4);
    \draw[blue] (z4) -- (z5); 
    \draw[red] (z5) -- (z6);
    \draw[blue] (z6) -- (z1);
    
    \draw[red] (z3) -- (z6);
    \draw[blue, dashed] (z3) -- (z6);
    
    \draw[dashed] (z1) -- (z3);
    
    \draw[dashed] (z6) -- (z4);

    \pic [draw, -, "$\phi_2$", angle eccentricity=1.5, angle radius=15] {angle = z3--z1--z2};
    \pic [draw, -, "$\phi_1$", angle eccentricity=1.5, angle radius=15] {angle = z1--z3--z6};
    
    \pic [draw, -, "$\phi_1'$", angle eccentricity=1.5, angle radius=15] {angle = z4--z6--z3};
    \pic [draw, -, "$\phi_3$", angle eccentricity=1.5, angle radius=15] {angle = z6--z4--z5};
    
\end{tikzpicture}
\hspace{1em}
\begin{tikzpicture}[scale=0.9]
    
    \vertex[label=left:$z_1$](z1) at (0, 0) {};
    \vertex[label=left:$z_2$](z2) at (0, 1.5) {};
    \vertex[label=above:$z_3$](z3) at (1.5, 0.75) {};
    \vertex[label=right:$z_4$](z4) at (3, 0) {};
    \vertex[label=right:$z_5$](z5) at (3, 1.5) {};
    
    \draw[red] (z1) -- (z2);
    \draw[blue] (z1) -- (z3);
    \draw[blue] (z2) -- (z3);
    
    \draw[blue] (z4) -- (z5);
    \draw[red] (z3) -- (z4);
    \draw[red] (z3) -- (z5);
    
    \draw[red] (z1) -- (z4);
    \draw[blue, dashed] (z1) -- (z4);
    
    \pic [draw, <-, "$\phi$", angle eccentricity=1.5, angle radius=12] {angle = z5--z4--z3};
    
\end{tikzpicture}
\hspace{1em}
\begin{tikzpicture}[scale=0.9]
    
    \vertex[label=below:$z_1$](z1) at (0, 0) {};
    
    \vertex[label=below:$z_2$](z2) at (3, 0) {};
    \vertex[label=above:$z_3$](z3) at (0, 1) {};
    
    \vertex[label=below:$z_4$](z4) at (1.5, -1) {};
    \vertex[label=above:$z_5$](z5) at (2, 1) {};
    
    \draw[dashed] (z1) -- node[below] {$R$} ++ (z2);
    
    \draw[red] (z1) -- (z3);
    \draw[blue] (z3) -- (z2);
    
    \draw[blue] (z1) -- (z4); 
    \draw[red, dashed] (z1) -- (z4); 
    
    \draw[red] (z4) -- (z2);
    \draw[blue, dashed] (z4) -- (z2);
    
    \draw[blue] (z1) -- (z5);
    \draw[red] (z2) -- (z5);
    
    \pic [draw, ->, "$\phi$", angle eccentricity=1.5, angle radius=20] {angle = z2--z1--z5};
    
\end{tikzpicture}
\caption{Labels of points involved in the three types of pairs of linked 2-changes. Left: type 0. Center: type 1a. Right: type 1b.\label{fig:2change_pairs}}
\end{figure}

\begin{lemma}\label{lemma:type0_2change}
    Let $\Delta^{\mathrm{link}}_\mathrm{min}$ denote the minimum improvement
    of any type 0 pair of linked 2-changes, and assume that $\X \subseteq [-D, D]^d$. Then
    \[
        \prob(\Delta^\mathrm{link}_{\mathrm{min}} \in (0, \epsilon])
            = O\left(
                \frac{d D^2 n^6 \epsilon^2}{\sigma^4}
            \right).
    \]
\end{lemma}

\begin{proof}
    The result follows from the independence of $\Delta_1$ and
    $\Delta_2$ when conditioning on $A_1 = a_1$. Observe
    that $\prob(\Delta^\mathrm{link} \in (0, \epsilon])
    \leq \prob(\Delta_1 \in (0, \epsilon] \wedge \Delta_2 \in (0,\epsilon])$.
    Thus, using \Cref{lemma:single_2change_one_distance3},
    \begin{align*}
        \prob(\Delta^\mathrm{link} \in (0,\epsilon] \given A_1 = a_1)
             = O\left(
                \left(
                    \frac{\sqrt{d}D}{\sigma^2} + \frac{d^{3/4}\sqrt{D}}{\sigma \sqrt{a_1}}
                \right)^2
                \epsilon^2
            \right).
    \end{align*}
    Straightforward algebra yields
    \[
                \left(
                    \frac{\sqrt{d}D}{\sigma^2} + \frac{d^{3/4}\sqrt{D}}{\sigma \sqrt{a_1}}
                \right)^2
            = O\left(
                \frac{dD^2}{\sigma^4} + \frac{d^{3/2}D}{\sigma^2 a_1} + \frac{d^{5/4}D^{3/2}}{\sigma^3\sqrt{a_1}}
            \right).
    \]
    Using \Cref{lemma:chi_stochdom,lemma:chi_expect} to integrate out $a_1$,
    we obtain
    \[
        \frac{dD^2}{\sigma^4} + \frac{dD}{\sigma^3} + \frac{dD^{3/2}}{\sigma^{7/2}}
            = O\left(
                \frac{dD^2}{\sigma^4}
            \right).
    \]
    Taking a union bound over the $O(n^6)$ different type 0 pairs completes
    the proof.
\end{proof}

\subsection{Type 1}\label{sec:type1}

As mentioned previously, type 1 linked pairs can be subdivided into two
distinct subtypes. Subtype 1a shares exactly one edge between the two 2-changes,
while subtype 1b shares two edges.

\subsubsection{Type 1a}

We first consider type 1a. See \Cref{fig:2change_pairs} (center) for
a graphical representation of the type, as well as the labels of the
points and edges involved.

Let the 2-change replacing $\{z_1, z_2\}$ and $\{z_3, z_4\}$ by
$\{z_2, z_3\}$ and $\{z_1, z_4\}$ be called $S_1$, and the 2-change
replacing $\{z_1, z_4\}$ and $\{z_3, z_5\}$ by $\{z_1, z_3\}$ and
$\{z_4, z_5\}$ be called $S_2$.

We proceed by conditioning on $A_2 = \|z_3 - z_4\| = a_2$
and $A_3 = \|z_4 - z_5\| = a_3$. Using
\Cref{lemma:single_2change_conditioned}, we can then compute the
probability that $\Delta_1 \in (0, \epsilon]$. Moreover,
the location of $z_5$ is then still random.
Hence, the random variable
$\eta = \|z_3 - z_5\| - \|z_4 - z_5\|$ can be analyzed
independently from $\Delta_1$. 

For the density of $\eta$, we
have the following lemma from Englert et al \cite{englertWorstCaseProbabilistic2014}.

\begin{lemma}[\mbox{\cite[Lemma 15, modified]{englertWorstCaseProbabilistic2014}}]\label{lemma:feta_bound_englert}
    Let $i \in [2]$, and assume that $\X \subseteq [-D, D]^d$. For $a_2, a_3 \in (0, 2\sqrt{d}D]$ and $\eta \in (-a_2, \min\{a_2, 2a_3 - a_2\})$,
    \begin{align*}
        f_{\eta|A_2 = a_2, A_3=a_3}(\eta) \leq M_{\phi} \cdot \begin{cases}
            \sqrt{\frac{2}{a_2^2 - \eta^2}}, & \text{if } a_3 \geq a_2, \\
            \sqrt{\frac{2}{(a_2 + \eta)(2a_3 - a_2 - \eta}}, & \text{if } a_3 < a_2,
        \end{cases}
    \end{align*}
    where $M_\phi = \max_{0 \leq \phi \leq \pi} f_{\phi|A_2=a_2,A_3=a_3}(\phi)$.
    For $\eta \notin (-r, \min\{a_2, 2a_3 - a_2\})$, the density vanishes.
\end{lemma}

Note that the factor $M_{\phi}$ was not present in the original statement of
\Cref{lemma:feta_bound_englert}. This is because the original statement concerned a simplified random experiment,
wherein the points $z_5$ and $z_3$ are chosen uniformly from a hyperball
centered on $z_4$. As such, $\phi$ is assumed to be distributed uniformly\footnote{
This assumption is only valid for $d = 2$. 
To see this, observe that by conditioning on $A_i = a_i$, the point $z_i$
is distributed uniformly on the $(d-1)$-sphere with radius $a_i$.
For $d > 2$, the density
of $\phi$ is thus concentrated near $\phi = \pi/2$. An upper bound for this
density can be obtained by setting $s = 0$ in \Cref{thm:angle_distr_bound},
yielding $O(\sqrt{d})$. As Englert et al.\ assume $d$ to be constant,
this has no effect on their eventual result.}.
Since we do not analyze a simplified random experiment, we cannot
make this assumption. However, examining the original proof of
\Cref{lemma:feta_bound_englert}, this can be resolved by
simply inserting the upper bound of the density of $\phi$, conditioned
on $A_2 = a_2$ and $A_3 = a_3$.
This bound is provided to us by \Cref{cor:angle_distr_bound}.

\begin{lemma}\label{lemma:type1a_2change_2}
    Let $\Delta_2$ be the improvement yielded by $S_2$, and assume that $\X \subseteq [-D, D]^d$. Then
        \[
            \prob(\Delta_2 \in (0, \epsilon] \given A_2 = a_2) = O\left(
                \left(
                    \frac{d^{1/4}\sqrt{D}}{\sigma} + \sqrt{\frac{d}{a_2}} 
                \right) \cdot \sqrt{\epsilon}
            \right).
        \]
\end{lemma}

\begin{proof}
    We obtain the density of $\eta$ from \Cref{lemma:feta_bound_englert}.
    As before, we need to subdivide
    into the cases $a_2 \leq a_3$ and $a_2 \geq a_3$.
    ~\paragraph*{Case 1: $a_3 \leq a_2$.} For this case, the conditional
    density of $\eta$ reads
    \[
        f_{\eta|A_2=a_2, A_3=a_3}(\eta) \leq M_{\phi} \cdot \begin{cases}
            \sqrt{\frac{2}{a_3(a_2 + \eta)}}, & \eta \leq a_3 - a_2, \\
            \sqrt{\frac{2}{a_3(2a_3 - a_2 - \eta}}, & \eta \geq a_3 - a_2.
        \end{cases}
    \]
    We assume that the random variable $\|z_1 - z_4\| - \|z_1 - z_3\|$ has
    been fixed by the adversary. This fixes an interval of size $\epsilon$ for
    $\eta$ to fall within, should $\Delta_2 \in (0,\epsilon]$ occur. Observe
    that $f_{\eta|A_2=a_2,A_3=a_3}$ integrated over any
    interval of size $\epsilon$ yields at most $O(M_\phi\sqrt{\epsilon/a_3})$.
    Since $a_3 \leq a_2$, we have $M_{\phi} = O(\sqrt{d} + d^{1/4}\sqrt{Da_3}/\sigma)$.
    Thus, for any interval $I$ of size $\epsilon$,
    \[
        \prob(\eta \in I \given A_2 = a_2, A_3 = a_3)
            = O\left(
                \left(\sqrt{\frac{d}{a_3}} + \frac{d^{1/4}\sqrt{D}}{\sigma}\right)
                    \cdot \sqrt{\epsilon}
            \right).
    \]
    ~\paragraph*{Case 2: $a_3 \geq a_2$.} For this case, we have
    \[
        f_{\eta|A_2 = a_2, A_3 = a_3}(\eta) = M_\phi \sqrt{\frac{2}{a_2}}
            \cdot \sqrt{\frac{1}{a_2 - |\eta|}}.
    \]
    Similarly as in Case 1, this function integrates to at most
    $O(M_\phi\sqrt{\epsilon/a_2})$. Here, we have $M_{\phi} = O(\sqrt{d}
    + d^{1/4}\sqrt{Da_2}/\sigma)$, so we obtain
    \[
        \prob(\eta \in I \given A_2 = a_2, A_3 = a_3)
            = O\left(
                \left(\sqrt{\frac{d}{a_2}} + \frac{d^{1/4}\sqrt{D}}{\sigma}\right)
                    \cdot \sqrt{\epsilon}
            \right).
    \]
    
    Combining the two cases above, we see that
    \[
        \prob(\Delta_2 \in (0,\epsilon]\given A_2 = a_2, A_3 = a_3)
            = O\left(
                \left(
                    \frac{d^{1/4}\sqrt{D}}{\sigma} + \sqrt{\frac{d}{a_2}}
                        + \sqrt{\frac{d}{a_3}}
                \right) \cdot \sqrt{\epsilon}
            \right).
    \]
    We can now integrate out $a_3$ using
    \Cref{lemma:chi_expect,lemma:chi_stochdom}. Then, using $D \geq 1$, $d \geq 2$ and
    $\sigma \leq 1$, we eventually arrive at the stated result.
\end{proof}

Using \Cref{lemma:chi_expect,lemma:type1a_2change_2},
we can easily prove the following statement about type 1a pairs of 2-changes.

\begin{lemma}\label{lemma:type1a_2change}  
    Let $\Delta^{\mathrm{link}}_\mathrm{min}$ denote the minimum improvement
    of any type 1a pair of 2-changes, and assume that $\X \subseteq [-D, D]^d$. Then
    \[
        \prob(\Delta^{\mathrm{link}}_\mathrm{min} \in (0,\epsilon])
            = O\left(
                \frac{n^5 d^{3/4}D^{3/2}}{\sigma^{3}} \epsilon^{3/2}
            \right).
    \]
\end{lemma}

\begin{proof}
    As in the proof of \Cref{lemma:type0_2change}, we can simply
    use \Cref{lemma:single_2change_one_distance3,lemma:type1a_2change_2}
    to compute the probability that both $\Delta_1 \in (0, \epsilon]$
    and $\Delta_2 \in (0, \epsilon]$, which bounds the probability
    that $\Delta_1 + \Delta_2 \in (0,\epsilon]$:
    \begin{align*}
        \prob(\Delta_1, \Delta_2 \in (0, \epsilon] \given A_2 = a_2)
            &= O\left(
                \left(
                    \frac{d^{3/4}D^{3/2}}{\sigma^3} + \frac{dD}{\sigma^2\sqrt{a_2}}
                        + \frac{d^{5/4}\sqrt{D}}{\sigma a_2}
                \right)\cdot \epsilon^{3/2}
            \right).
    \end{align*}
    Using \Cref{lemma:chi_expect,lemma:chi_stochdom}, with $d \geq 2$, $D \geq 1$ and
    $\sigma \leq 1$ in conjunction with a union bound
    over the $O(n^5)$ pairs of type 1a yields the result.
\end{proof}

\subsubsection{Type 1b}\label{sec:type1b}

The final type of linked pair we consider is type 1b.
See \Cref{fig:2change_pairs} (right) for a graphical representation.

Let $S_1$ denote the 2-change replacing $\{z_1, z_3\}$ and $\{z_2, z_4\}$
with $\{z_2, z_3\}$ and $\{z_1, z_4\}$, and let $S_2$ denote the 2-change
replacing $\{z_2, z_5\}$ and $\{z_1, z_4\}$ with $\{z_1, z_5\}$ and $\{z_2, z_5\}$.
From \Cref{fig:2change_pairs}, it is evident that we can
treat $\Delta_1$ and $\eta = \|z_2 - z_5\| - \|z_1 - z_5\|$ as independent
variables, as long as we condition on $R = r$.

\begin{lemma}\label{lemma:type1b_2change}
    Let $\Delta^{\mathrm{link}}_\mathrm{min}$ denote the minimum improvement
    of any type 1b pair of 2-changes, and assume that $\X \subseteq [-D, D]^d$. Then
    \[
        \prob(\Delta^{\mathrm{link}}_\mathrm{min} \in (0,\epsilon])
            = O\left(
                \frac{n^5 d^{3/4}D^{3/2}}{\sigma^{3}} \epsilon^{3/2}
            \right).
    \]
\end{lemma}

\begin{proof}
    The proof follows along the exact same lines as \Cref{lemma:type1a_2change}.
    small modifications.
\end{proof}

\Cref{lemma:type0_2change,lemma:type1a_2change,lemma:type1b_2change} enable us to prove
an upper bound to the smoothed complexity of 2-opt in the present
probabilistic model.

\begin{theorem}\label{thm:smoothed_complexity}
    The expected number of iterations performed by 2-opt for smoothed Euclidean instances
    of TSP in $d \geq 2$ dimensions is bounded from above by
    $
        O\left(
            d D^2 n^{4 + \frac{1}{3}}/\sigma^2
        \right)
    $.
\end{theorem}

\begin{proof}
    We assume for this proof that the entire instance is contained
    within $[-D, D]^d$, with $D = \Theta(1 + \sigma\sqrt{n \log n})$.
    This occurs with probability at least $1 - 1/n!$. Thus, with
    probability at least $1 - 1/n!$, the longest tour in the instance
    has length at most $2\sqrt{d}Dn$.
    The assumption that the entire instance lies within
    this hypercube enables us to use
    \Cref{lemma:type0_2change,lemma:type1a_2change,lemma:type1b_2change},
    which were proved under this assumption.
    
    Let $E$ denote the event that, among all type 0 and type 1 linked pairs of 2-changes,
    the pair with the smallest
    improvement is of type 0, and let $E^c$ denote the event that this pair is
    of type 1a or type 1b. Let the random variable $T$ denote the number of
    iterations taken by 2-opt to reach a local optimum.
    
    We first compute $\expect(T\given E)$. We apply \Cref{lemma:blackbox}
    with $\alpha = 2$, which
    is feasible due to \Cref{lemma:type0_2change}. We then obtain
    immediately that $\expect(T\given E) = O(dD^2n^4/\sigma^2)$.
    
    Next, we compute $\expect(T\given E^c)$. In this case, we apply
    \Cref{lemma:blackbox} with $\alpha = 3/2$
    (cf.\ \Cref{lemma:type1a_2change,lemma:type1b_2change}). This yields
    $\expect(T\given E^c) = O(dD^2n^{4+\frac{1}{3}}/\sigma^2)$.
    
    Combining the bounds for $E$ and $E^c$ yields the result.
\end{proof}

\section{\boldmath Improving the Analysis for \texorpdfstring{$d \geq 3$}{TEXT}}
\label{sec:d3}

The bottleneck in \Cref{thm:smoothed_complexity} stems from
\Cref{lemma:type1a_2change,lemma:type1b_2change}, which bound
the probability that any linked pair of type 1a or type 1b improves
the tour by at most $\epsilon$. The probability given by these lemmas
is proportional to $\epsilon^{3/2}$, which yields an extra factor
of $n^{1/3}$ compared to type 0 linked pairs.

For $d \geq 3$, we can improve this to $\epsilon^2$, yielding improved
smoothed complexity bounds. The key to this improvement is to use
the second part of \Cref{cor:angle_distr_bound} to bound
the density of $\eta_i$ as in \Cref{lemma:eta_density}.
This immediately yields the following result
on $\eta_i = \|a - z_i\| - \|b - z_i\|$.

\begin{lemma}\label{lemma:eta_density_d3}
    Let $i \in [2]$, and assume that $\X \subseteq [-D, D]^d$. The density of $\eta_i$ in $d \geq 3$ dimensions, conditioned on $A_i = a_i$ and
    $R = r$, is bounded from above by
    \[
        O\left(\frac{a_i + r}{a_i r} \cdot \left(\sqrt{d} + \frac{D\min\{r, a_i\}}{\sigma^2}\right)\right).
    \]
\end{lemma}

\begin{proof}
    We call the desired density $f_{\eta_i|A=a_i,R=r}$. From
    \Cref{lemma:eta_density}, we know that
    \[
        f_{\eta_i|A_i=a_i,R=r}(\eta) \leq \frac{a_i+r}{a_i r}
            \cdot \frac{f_{\phi_i|A_i=a_i,R=r}(\phi_i(\eta))}{|\sin \phi_i(\eta)|}.
    \]
    Since $d \geq 3$, we can use the second part of \Cref{cor:angle_distr_bound}
    to obtain the desired bound, making use of the
    assumption that all points fall within $[-D,D]^d$.
\end{proof}

\Cref{lemma:eta_density_d3} enables us to find an improved version of \Cref{lemma:single_2change_conditioned}.

\begin{lemma}\label{lemma:single_2change_d3}
    Let $\Delta$ denote the improvement of a 2-change in $d \geq 3$ dimensions. Let
    $i \in [2]$, and assume that $\X \subseteq [-D, D]^d$. Then
    \[
        \prob(\Delta \in (0, \epsilon] \given A_i = a_i, R = r) = O\left(
            \left(\frac{\sqrt{d}}{\min\{a_i, r\}} + \frac{D}{\sigma^2}\right) \cdot \epsilon
        \right).
    \]
\end{lemma}

\begin{proof}
    Let $j = 3 - i$. We assume that $\eta_j = t$ is fixed by the adversary. Then
    $\Delta \in (0, \epsilon]$ iff $\eta_i \in (-t, -t + \epsilon] =: I$, an interval
    of size $\epsilon$. By \Cref{lemma:eta_density_d3}, we have a bound for
    the density of $\eta_i$. Thus, we find
    \[
        \prob(\Delta \in (0, \epsilon] \given A_i = a_i, R=r)
            = O\left(\frac{a_i + r}{a_i r} \cdot \left(\sqrt{d} + D\min\{r, a_i\}/\sigma^2\right) \cdot \epsilon\right).
    \]
    Considering the cases $a_i \leq r$ and $a_i < r$ separately and using the assumptions that
    all points lie within $[-D, D]^d$ and that
    $D \geq 1$ and $\sigma \leq 1$ yields the stated result.
\end{proof}

The following lemma now yields the probability that any linked pair of 2-changes
improves the tour by at most $\epsilon$. We omit the proof, since it follow easily
from \Cref{lemma:single_2change_d3} along the same lines as
the lemmas in \Cref{sec:linked_pairs}.

\begin{lemma}\label{lemma:linked_pairs_d3}
    Let $\Delta^{\mathrm{link}}_\mathrm{min}$ denote the minimum improvement
    of any linked pair of 2-changes of type 0 or type 1 for $d \geq 3$,
    and assume that $\X \subseteq [-D, D]^d$. Then
    \[
        \prob(\Delta^\mathrm{link}_{\mathrm{min}} \in (0, \epsilon])
            = O\left(
                \frac{D^2 n^6 \epsilon^2}{\sigma^4}
            \right).
    \]
\end{lemma}

We then obtain our result for $ d \geq 3$.

\begin{theorem}\label{thm:smoothed_complexity_d3}
    The expected number of iterations performed by 2-opt for smoothed Euclidean instances
    of TSP in $d \geq 3$ dimensions is bounded from above by
    $
        O\left(
            \sqrt{d}D^2 n^4/\sigma^2
        \right)
    $.
\end{theorem}

\begin{proof}
    The theorem follows immediately from applying \Cref{lemma:blackbox,lemma:linked_pairs_d3},
    since by \Cref{lemma:box} any tour in our smoothed instance has length at most
    $2\sqrt{d}Dn$ with probability at least $1 - 1/n!$.
\end{proof}

\section{Discussion}

\begin{table}
\centering
\begin{tabular}{lccc}
             & Englert, R\"oglin \& V\"ocking \cite{englertWorstCaseProbabilistic2014} & Manthey \& Veenstra \cite{mantheySmoothedAnalysis2Opt2013} & This paper \\ \hline\hline
     $d = 2$ & $O\left(n^{4+\frac{1}{3}}/\sigma^{5+\frac{1}{3}}\cdot\log\frac{n}{\sigma}\right)$ & -
         & $O\left(n^{4+\frac{1}{3}}/\sigma^2\right)$ \\ \hline
     $d = 3$ & $O\left(n^{4+\frac{1}{3}}/\sigma^8 \cdot \log \frac{n}{\sigma}\right) $ & -
         & $O\left(n^4/\sigma^2 \right)$ \\ \hline
     $d \geq 4$ & $O\left(c_d \cdot n^{4+\frac{1}{3}}/\sigma^{8d/3}\right)$ & $O\left(\sqrt{d}n^4/\sigma^4 \right)$
         & $O\left(\sqrt{d}n^4/\sigma^2\right)$
\end{tabular}
\caption{Previous and current smoothed complexity bounds for Gaussian noise, for $\sigma = O(1/\sqrt{n\log n})$.
Note that for $d \geq 4$, the bounds of Englert et al.\ include a factor $c_d$ which is super-exponential in $d$.
\label{table_smallsigma}}
\end{table}

\begin{table}
\centering
\begin{tabular}{lccc}
             & Englert, R\"oglin \& V\"ocking \cite{englertWorstCaseProbabilistic2014} & Manthey \& Veenstra \cite{mantheySmoothedAnalysis2Opt2013} & This paper \\ \hline\hline
     $d = 2$ & $O\left(n^7 \log^{3+\frac{2}{3}} n \right)$ & -
         & $O\left(n^{5+\frac{1}{3}}\log n\right)$ \\ \hline
     $d = 3$ & $O\left(n^{8 + \frac{1}{3}}\log^5 n \right) $ & -
         & $O\left(n^5\log n \right)$ \\ \hline
     $d \geq 4$ & $O\left(c_d \cdot n^{4+\frac{1 + 4d}{3}} \log^{1 + \frac{4d}{3}}n\right)$ & $O\left(\sqrt{d}n^6\log^2 n \right)$
         & $O\left(\sqrt{d}n^{5} \log n\right)$
\end{tabular}
\caption{Previous and current smoothed complexity bounds for Gaussian noise, for $\sigma = \Omega(1/\sqrt{n\log n})$.
Note that for $d \geq 4$, the bounds of Englert et al.\ include a factor $c_d$ which is super-exponential in $d$.
\label{table_largesigma}}
\end{table}

For convenience, we provide comparisons of the previous smoothed complexity
bounds with our bound from \Cref{thm:smoothed_complexity} in
\Cref{table_smallsigma,table_largesigma}. These bounds
are provided both for small values of $\sigma$ and for large values, meaning
$\sigma = \Omega(1/\sqrt{n\log n})$ and $\sigma = O(1/\sqrt{n \log n})$.

Observe from \Cref{table_largesigma,table_smallsigma} that the bound for
$d = 2$ has a worse dependence on $n$ compared to the bound for $d \geq 3$. The
technical reasons for this difference can be understood from \Cref{sec:d3}.
A more intuitive explanation for the difference is that our analysis benefits
from large angles between edges in the smoothed TSP instance. In $d = 2$,
the density of these angles is maximal when they are small, while
for $d \geq 3$ it is maximal when the angles are large. In effect, this means that the
adversary has less power to specify these angles to our detriment when
$d \geq 3$.

From these tables, the greatest improvement is made
for $d = 3$, where we improve by $n^{3+\frac{1}{3}}\log^{4}n$ in
the large $\sigma$ case, and by $\sqrt[3]{n}\log(n/\sigma)/\sigma^6$ for small $\sigma$.
For $d = 2$, the improvement is more
modest at $n^{1+\frac{2}{3}}\log^{2+\frac{2}{3}}n$ for large $\sigma$
and $\log(n/\sigma)/\sigma^{3 + \frac{1}{3}}$
for small $\sigma$.
For $d \geq 4$, we improve by $n\log n$
for large $\sigma$,
and by $\sigma^{-2}$ for small $\sigma$.

Note that we improve upon previous
bounds mainly in the dependence on the perturbation strength. In an intuitive sense, this
is most substantial for instances that are
weakly perturbed from the adversarial instance, or in other words,
that are close to worst case. In addition, the small-$\sigma$ case
is considered more interesting for a smoothed analysis,
since small $\sigma$ model the intuition of smoothed analysis of a small perturbation, 
while large $\sigma$ reduce the analysis basically to an average-case analysis
In order to improve the explicit dependence on $n$,
which is the same as for Manthey \& Veenstra \cite{mantheySmoothedAnalysis2Opt2013}, we
believe new techniques are necessary.

As a final comment, we note that the techniques we employed in 
\Cref{sec:single_2change,sec:d3}
can also be used to improve and significantly simplify the analysis of the
one-step model used by Englert et al \cite{englertWorstCaseProbabilistic2014}.
For $d \geq 3$, the improvement amounts to a factor of $n^{1/3}\phi^{1/6}\log(n\phi)$,
while for $d = 2$, the improvement is just $\log(n\phi)$,
where $\phi$ denotes the upper bound of the density functions
used in the one-step model.

\printbibliography

@book{aartsLocalSearchCombinatorial2003,
  title = {Local {{Search}} in {{Combinatorial Optimization}}},
  editor = {Aarts, Emile and Lenstra, Jan Karel},
  year = {2003},
  publisher = {{Princeton University Press}},
  doi = {10.2307/j.ctv346t9c},
  urldate = {2021-09-03}
}

@book{abramowitzHandbookMathematicalFunctions1974,
  title = {Handbook of {{Mathematical Functions}}, {{With Formulas}}, {{Graphs}}, and {{Mathematical Tables}},},
  author = {Abramowitz, Milton},
  year = {1974},
  publisher = {{Dover Publications, Inc.}},
  address = {{USA}},
  isbn = {978-0-486-61272-0}
}

@article{amosComputationModifiedBessel1974,
  title = {Computation of Modified {{Bessel}} Functions and Their Ratios},
  author = {Amos, D. E.},
  year = {1974},
  journal = {Mathematics of Computation},
  volume = {28},
  number = {125},
  pages = {239--251},
  issn = {0025-5718, 1088-6842},
  doi = {10.1090/S0025-5718-1974-0333287-7},
  urldate = {2022-09-06},
  langid = {english}
}

@article{apostolElementaryViewEuler1999,
  title = {An {{Elementary View}} of {{Euler}}'s {{Summation Formula}}},
  author = {Apostol, Tom M.},
  year = {1999},
  journal = {The American Mathematical Monthly},
  volume = {106},
  number = {5},
  pages = {409--418},
  publisher = {{Mathematical Association of America}},
  issn = {0002-9890},
  doi = {10.2307/2589145},
  urldate = {2022-09-30}
}

@article{chandraNewResultsOld1999,
  title = {New {{Results}} on the {{Old}} K-Opt {{Algorithm}} for the {{Traveling Salesman Problem}}},
  author = {Chandra, Barun and Karloff, Howard and Tovey, Craig},
  year = {1999},
  month = jan,
  journal = {SIAM Journal on Computing},
  volume = {28},
  number = {6},
  pages = {1998--2029},
  publisher = {{Society for Industrial and Applied Mathematics}},
  issn = {0097-5397},
  doi = {10.1137/S0097539793251244},
  urldate = {2022-07-07}
}

@article{engelsAveragecaseApproximationRatio2009,
  title = {Average-Case Approximation Ratio of the 2-Opt Algorithm for the {{TSP}}},
  author = {Engels, Christian and Manthey, Bodo},
  year = {2009},
  month = mar,
  journal = {Operations Research Letters},
  volume = {37},
  number = {2},
  pages = {83--84},
  issn = {0167-6377},
  doi = {10.1016/j.orl.2008.12.002},
  urldate = {2022-01-18},
  langid = {english}
}

@article{englertSmoothedAnalysis2Opt2016,
  title = {Smoothed {{Analysis}} of the 2-{{Opt Algorithm}} for the {{General TSP}}},
  author = {Englert, Matthias and R{\"o}glin, Heiko and V{\"o}cking, Berthold},
  year = {2016},
  month = sep,
  journal = {ACM Transactions on Algorithms},
  volume = {13},
  number = {1},
  pages = {10:1--10:15},
  issn = {1549-6325},
  doi = {10.1145/2972953},
  urldate = {2021-08-19}
}

@article{englertWorstCaseProbabilistic2014,
  title = {Worst {{Case}} and {{Probabilistic Analysis}} of the 2-{{Opt Algorithm}} for the {{TSP}}},
  author = {Englert, Matthias and R{\"o}glin, Heiko and V{\"o}cking, Berthold},
  year = {2014},
  month = jan,
  journal = {Algorithmica},
  volume = {68},
  number = {1},
  pages = {190--264},
  issn = {1432-0541},
  doi = {10.1007/s00453-013-9801-4},
  urldate = {2022-06-30},
  langid = {english},
  note = {Corrected version: \href{https://arxiv.org/abs/2302.06889}{https://arxiv.org/abs/2302.06889}}
}

@book{johnsonContinuousUnivariateDistributions1995,
  title = {Continuous {{Univariate Distributions}}, {{Volume}} 2},
  author = {Johnson, Norman L. and Kotz, Samuel and Balakrishnan, Narayanaswamy},
  year = {1995},
  month = may,
  publisher = {{John Wiley \& Sons}},
  googlebooks = {BTANEAAAQBAJ},
  isbn = {978-0-471-58494-0},
  langid = {english}
}

@book{korteCombinatorialOptimizationTheory2000,
  title = {Combinatorial {{Optimization}}: {{Theory}} and {{Algorithms}}},
  shorttitle = {Combinatorial {{Optimization}}},
  author = {Korte, Bernhard and Vygen, Jens},
  year = {2000},
  series = {Algorithms and {{Combinatorics}}},
  publisher = {{Springer-Verlag}},
  address = {{Berlin Heidelberg}},
  doi = {10.1007/978-3-662-21708-5},
  urldate = {2021-09-03},
  isbn = {978-3-662-21708-5},
  langid = {english}
}

@inproceedings{mantheySmoothedAnalysis2Opt2013,
  title = {Smoothed {{Analysis}} of the 2-{{Opt Heuristic}} for the {{TSP}}: {{Polynomial Bounds}} for {{Gaussian Noise}}},
  shorttitle = {Smoothed {{Analysis}} of the 2-{{Opt Heuristic}} for the {{TSP}}},
  booktitle = {Algorithms and {{Computation}}},
  author = {Manthey, Bodo and Veenstra, Rianne},
  editor = {Cai, Leizhen and Cheng, Siu-Wing and Lam, Tak-Wah},
  year = {2013},
  series = {Lecture {{Notes}} in {{Computer Science}}},
  pages = {579--589},
  publisher = {{Springer}},
  address = {{Berlin, Heidelberg}},
  doi = {10.1007/978-3-642-45030-3_54},
  isbn = {978-3-642-45030-3},
  langid = {english},
  note = {Full, improved version: https://arxiv.org/abs/2308.00306}
}

@article{mantheySmoothedAnalysisAnalysis2011,
  title = {Smoothed {{Analysis}}: {{Analysis}} of {{Algorithms Beyond Worst Case}}},
  shorttitle = {Smoothed {{Analysis}}},
  author = {Manthey, Bodo and R{\"o}glin, Heiko},
  year = {2011},
  month = dec,
  journal = {it - Information Technology},
  volume = {53},
  number = {6},
  pages = {280--286},
  publisher = {{De Gruyter Oldenbourg}},
  doi = {10.1524/itit.2011.0654},
  urldate = {2022-10-07},
  langid = {english}
}

@incollection{mantheySmoothedAnalysisLocal2021,
  title = {Smoothed {{Analysis}} of {{Local Search}}},
  booktitle = {Beyond the {{Worst-Case Analysis}} of {{Algorithms}}},
  author = {Manthey, Bodo},
  editor = {Roughgarden, Tim},
  year = {2021},
  pages = {285--308},
  publisher = {{Cambridge University Press}},
  address = {{Cambridge}},
  doi = {10.1017/9781108637435.018},
  urldate = {2022-10-07},
  isbn = {978-1-108-49431-1}
}

@article{papadimitriouEuclideanTravellingSalesman1977,
  title = {The {{Euclidean}} Travelling Salesman Problem Is {{NP-complete}}},
  author = {Papadimitriou, Christos H.},
  year = {1977},
  month = jun,
  journal = {Theoretical Computer Science},
  volume = {4},
  number = {3},
  pages = {237--244},
  issn = {0304-3975},
  doi = {10.1016/0304-3975(77)90012-3},
  urldate = {2022-10-03},
  langid = {english}
}

@article{spielmanSmoothedAnalysisAlgorithms2004,
  title = {Smoothed Analysis of Algorithms: {{Why}} the Simplex Algorithm Usually Takes Polynomial Time},
  shorttitle = {Smoothed Analysis of Algorithms},
  author = {Spielman, Daniel A. and Teng, Shang-Hua},
  year = {2004},
  month = may,
  journal = {Journal of the ACM},
  volume = {51},
  number = {3},
  pages = {385--463},
  issn = {0004-5411},
  doi = {10.1145/990308.990310},
  urldate = {2022-10-03}
}

@article{spielmanSmoothedAnalysisAttempt2009,
  title = {Smoothed Analysis: An Attempt to Explain the Behavior of Algorithms in Practice},
  shorttitle = {Smoothed Analysis},
  author = {Spielman, Daniel A. and Teng, Shang-Hua},
  year = {2009},
  month = oct,
  journal = {Communications of the ACM},
  volume = {52},
  number = {10},
  pages = {76--84},
  issn = {0001-0782, 1557-7317},
  doi = {10.1145/1562764.1562785},
  urldate = {2022-10-07},
  langid = {english}
}

\end{document}